%% file: main.tex
\documentclass{article}
\usepackage[margin=1in]{geometry}

\input{preamble}
\renewcommand{\emph}[1]{\textit{#1}}
\title{Online Combinatorial Optimization\\
with Graphical Dependencies}

\begin{document}

	\author{Zhimeng Gao\thanks{
	(zhimeng@gatech.edu)
    School of Computer Science,
        Georgia Tech.
        Supported in part by NSF awards CCF-2327010 and CCF-2440113.
    }
	\and Evangelia Gergatsouli\thanks{
         (evagerg@gmail.com)
         School of Computer Science,
        Georgia Tech.
    }
    \and Kalen Patton\thanks{
        (kpatton33@gatech.edu)
        School of Mathematics,
        Georgia Tech.
        Supported in part by NSF awards CCF-2327010 and CCF-2440113.
        }
	\and Sahil Singla\thanks{
        (ssingla@gatech.edu)
        School of Computer Science,
        Georgia Tech.
        Supported in part by NSF awards CCF-2327010 and CCF-2440113.
        }
}

\maketitle

\begin{abstract}
\medskip

Most existing work in online stochastic combinatorial optimization assumes that inputs are drawn from \textit{independent} distributions---a strong assumption that often fails in practice. At the other extreme, arbitrary correlations are equivalent to worst-case inputs via Yao’s minimax principle, making good algorithms often impossible. This motivates the study of intermediate models that capture mild correlations while still permitting nontrivial algorithms.

   \smallskip
   
In this paper, we study online combinatorial optimization under Markov Random Fields (MRFs), a well-established graphical model for structured dependencies. MRFs parameterize correlation strength via the maximum weighted degree $\Delta$, smoothly interpolating between independence ($\Delta = 0$) and full correlation ($\Delta \to \infty$). While na\"ively this yields  $e^{O(\Delta)}$-competitive algorithms and $\Omega(\Delta)$ hardness, we ask: when can we design tight $\Theta(\Delta)$-competitive algorithms?

   \smallskip
   
We present general techniques achieving $O(\Delta)$-competitive algorithms for both minimization and maximization problems under MRF-distributed inputs. For minimization problems with coverage constraints (e.g., Facility Location and Steiner Tree), we reduce to the well-studied $p$-sample  model \cite{KPSSV-ITCS19,LattanziMVWZ-NeurIPS21,CorreaCFOT21}. For maximization problems (e.g., matchings and combinatorial auctions with XOS buyers), we extend the ``balanced prices" framework for online allocation problems \cite{DuttFeldKessLuci2020} to MRFs. 

\bigskip

\end{abstract}

{\small
 \setcounter{tocdepth}{1}
    \tableofcontents
 }

\newpage
\section{Introduction}

\input{intro}

\input{prelims}

\section{Minimization Problems with MRFs}\label{sec:minimization}
\input{minimization}

\section{Online Matchings and Combinatorial Auctions}\label{sec:maximization}
\input{maximization}

\input{lowerBounds}

\clearpage

\appendix
\input{appendix}

\input{facilityLocation}

\clearpage

\bibliography{bibliography}
\bibliographystyle{alpha}

\end{document}

%% file: preamble.tex
\usepackage{amsmath,ulem,amssymb,amsthm}
\usepackage{dsfont}
\usepackage{natbib}
\usepackage{graphicx} 
\usepackage{xcolor}
\usepackage[backref=page,linktocpage=true,breaklinks,colorlinks,citecolor=blue,linkcolor=BrickRed]{hyperref}
\usepackage{cleveref}
\usepackage{enumitem} 
\usepackage{float}
\usepackage{bm}
\usepackage{tikz}
\usetikzlibrary{shapes,positioning}
\usepackage{thmtools,thm-restate}
\usepackage[ruled,linesnumbered]{algorithm2e}
\usepackage[dvipsnames]{xcolor}

\newcommand{\opt}{{\rm OPT}}
\newcommand{\alg}{{\rm ALG}}

\newcommand{\cost}{{\rm cost}}

\def\Real{\mathbb{R}}

\renewcommand{\Re}{\Real}

\def\cA{\mathcal{A}}

\def\cD{\mathcal{D}}

\def\cF{\mathcal{F}}

\def\cI{\mathcal{I}}

\def\cM{\mathcal{M}}

\newcommand{\ind}[1]{\mathds{1}_{\{#1\}}}

\let\R\Real
\let\Z\Integer

\def\Rp{\Real_{\geq 0}}

\def\argmax{\operatornamewithlimits{arg\,max}}

\def\E{{\mathbb{E}}}		

\def\1{\mathds{1}}             
\allowdisplaybreaks[4] 
\def\rp#1{^{\overline{#1}}}		

\newtheorem{theorem}{Theorem}[section]
\newtheorem{claim}[theorem]{Claim}

\newtheorem{lemma}[theorem]{Lemma}
\newtheorem{corollary}[theorem]{Corollary}

\newtheorem{fact}[theorem]{Fact}
\theoremstyle{definition}

\newtheorem{definition}[theorem]{Definition}

\newcommand{\lp}{\left}
\renewcommand{\rp}{\right}

\newcommand{\hide}[1]{}

\definecolor{toc}{RGB}{13,55,174}	
\usepackage{hyperref}				
\hypersetup{
colorlinks=true,
citecolor=toc,
filecolor=black,
linkcolor=toc,
urlcolor=toc
}

\newcommand{\probName}{Subadditive Coverage}
\newcommand{\IGNORE}[1]{}

%% file: intro.tex
Online combinatorial optimization is a foundational area in theoretical computer science, where an algorithm must respond to a sequence of $n$ requests arriving one by one, without knowledge of future inputs. The algorithm must make irrevocable decisions that satisfy global combinatorial constraints while optimizing an objective. The standard benchmark is the offline optimum---the best solution achievable with full knowledge of the input sequence---and performance is measured by the \emph{competitive ratio}, defined as the worst-case ratio between the algorithm’s objective and that of the offline optimum (see, e.g., books \cite{BorodinE98,fiat1998online,BN-Book09,vaze2023online}).
However, under fully adversarial inputs, many important online problems do not admit algorithms with a small competitive ratio. For example, online maximization problems like matchings and combinatorial auctions face $\Omega(n)$ hardness, while minimization problems like Facility Location and Steiner Tree have known $\Omega(\log n/\log\!\log n)$ lower bounds~\cite{Fota2008,DPZ2021,MakoWaxm1991}.

To overcome worst-case inapproximability, much recent work in online combinatorial optimization has focused on stochastic input models. A central example is the \emph{prophet model}, where each request $i \in [n]$ is independently drawn from a known distribution $\mathcal{D}_i$, and the algorithm aims to compete with the expected offline optimum over the product distribution $\mathcal{D} = \mathcal{D}_1 \times \dots \times \mathcal{D}_n$. This model has led to $O(1)$-competitive algorithms for a wide range of problems, including online matchings \cite{ChawlaMS-EC10,gravin2019prophet,EzraFGT-EC20}, combinatorial auctions \cite{FeldmanGL15,DuttFeldKessLuci2020,CorreaC23}, Facility Location and Steiner Tree~\cite{GargGuptLeonSank2008,ArgyFrieGuptSeil2022}.

However, these guarantees critically rely on the independence assumption. In the presence of arbitrary correlations between inputs, such guarantees collapse. This failure is formalized by Yao’s minimax principle~\cite{Yao1977}, which shows that online algorithms for known correlated distributions are no easier than those for worst-case adversarial inputs. Since real-world data often exhibits non-adversarial dependencies---and arbitrary correlations lead back to worst-case hardness---this motivates the study of online stochastic optimization under mild, structured correlations.

   Several recent works have begun addressing stochastic inputs with structured correlations. These include models with pairwise independence~\cite{CaraGravLuWang2021,GuptaHKL-IPCO24,DughmiKP-STOC24}, linear dependencies~\cite{bateni2015revenue,ImmorlicaSW23}, and other restricted forms (see related work in \Cref{sec:related}). However, their correlation models do not provide a unified framework that smoothly interpolates between full independence and worst-case correlation. To bridge this gap, Cai and Oikonomou~\cite{CaiOiko2021} proposed using Markov Random Fields (MRFs)---a classic  graphical model for reasoning about high-dimensional stochastic inputs with controlled dependencies.

\paragraph{Markov Random Fields.}

    MRFs model dependencies among a collection of random variables using an undirected graph (or hypergraph), where each vertex corresponds to a random variable and edges capture conditional dependencies. The key parameter is the \emph{maximum weighted degree} $\Delta$, which quantifies the total influence any single variable experiences from its neighbors. These $\Delta$-MRF distributions
    naturally interpolate between product distributions ($\Delta = 0$) and highly correlated distributions ($\Delta \to \infty$). Originating in statistical mechanics \cite{SnelKind1980}, MRFs have been widely used across domains such as computer vision~\cite{BlakKohlRoth2011}, causal inference~\cite{Pearl}, sensor networks~\cite{ChenSuXionXiao2016,PerrSigeDaSiJaya2009}, mechanism design~\cite{BrusCaiDask2020}, and combinatorial classification~\cite{KleiTard2002}, owing to their ability to capture complex but localized dependencies.

    A central technical property of MRFs is that each variable behaves approximately independently up to an $e^{O(\Delta)}$ multiplicative distortion (\Cref{lem:mrf_conditioned_bound}). This implies that algorithms designed for independent inputs often extend to the MRF setting with an $e^{O(\Delta)}$ loss in approximation. Conversely, since MRFs with maximum weighted degree $\Delta$ can encode general correlations over a small number of variables, standard worst-case hardness implies $\Omega(\Delta)$ lower bounds. Taken together, these two facts raise a natural  question:
    \begin{quote}
        Can we design   linear-in-$\Delta$   competitive online algorithms under MRF-distributed inputs?
    \end{quote}
    We answer this question in the affirmative for broad classes of online minimization and maximization  problems.
    Rather than exploiting problem-specific techniques, we present general frameworks to exponentially improve over the na\"ive $e^{O(\Delta)}$-competitive algorithms.

    \subsection{Our  Results and High-Level Techniques}\label{sec:intro-results}

    \paragraph{Minimization Problems.} We focus on the broad class of \emph{online subadditive covering} problems, which include canonical settings like Facility Location, Steiner Tree, and Set Cover. In these problems, each arriving request must be covered under a global cost-minimization subadditive objective. Our main result is a general reduction from the MRF arrival model to the well-studied \emph{$p$-sample independent model} \cite{KPSSV-ITCS19,LattanziMVWZ-NeurIPS21,KNR2020,KNR-SODA22}.
    
    In the $p$-sample model, a set of requests is chosen adversarially but then a random $p$-fraction is independently revealed in advance to the algorithm as \emph{sample}, while the remaining $(1-p)$ fraction appear one by one in an adversarial order (\Cref{def:p-sample}). Most algorithms in this model are also \emph{monotone}---their cost only decreases if the sample set increases. We show that any monotone algorithm in the $p$-sample model implies an algorithm in the MRF model, which  enables us to leverage existing $p$-sample  algorithms. 

 \begin{theorem}[Informal \Cref{thm:full_reduction_minimization}]\label{thm:full_reduction_minimization_intro}
    For any online subadditive covering problem, if  there exists an \( O\lp( \alpha \cdot  \log \lp( 1/p \rp)\rp) \)-competitive \emph{monotone} algorithm  in the \( p \)-sample independent model, then  there exists an \( O( \alpha \cdot \Delta ) \)-competitive algorithm  in the MRF arrival model. {This reduction holds even if the algorithm doesn't know the MRF distribution and is only given a single $n$ dimensional sample from the  underlying  MRF.}
    \end{theorem}

As a concrete application, we obtain $O(\Delta)$-competitive algorithms for Facility Location and Steiner Tree, using prior $p$-sample algorithms or their small modifications \cite{ArgyFrieGuptSeil2022}. 
Notably, our result provides the first nontrivial approximation guarantees for these problems in the presence of structured correlations. Moreover, it is particularly interesting that we only require a single-sample from the underlying distribution, since learning MRF distributions is a challenging problem and an active area of research. Our theorem permits one to bypass learning the correlated MRF distribution to design an algorithm directly.

  \begin{corollary}\label{cor:log-guarantee-min}
        For online Facility Location and online Steiner Tree in the MRF arrival model with maximum weighted degree \( \Delta \), there exist $O(\Delta)$-competitive algorithms given  only a single $n$ dimensional sample from the  underlying  MRF.
    \end{corollary}

\noindent These results are nearly tight:  in \Cref{sec:lowerBounds} we show $\tilde{\Omega}(\Delta)$ hardness for both online Facility Location and online Steiner Tree by leveraging the fact that $\Delta$-MRF distributions can capture general correlations that are given by a Markov chain (\Cref{lem:mrf-mc-coupling}).

 We also remark that for online Set Cover, if the conjectured $O(\log(mn)\cdot \log(1/p))$-competitive guarantee in the $p$-sample model of \cite{GuptKehnLevi2024}  holds via a monotone algorithm, our reduction would yield an $O(\log(mn) \cdot \Delta)$-competitive algorithm under MRF arrivals.

At a high level, the proof of \Cref{thm:full_reduction_minimization_intro} is based on the intuition that MRF distributions behave independently, up to an $e^{O(\Delta)}$ multiplicative distortion. This allows us to view them in the $p$-sample model for $p=e^{-O(\Delta)}$, however, the formal reduction is much more subtle. To this end, we introduce an intermediate $(\frac12,p)$-sample model, where again the requests are chosen adversarially with a random fraction being revealed in advance (similar to the $p$-sample model), but there are correlations between the revelations of  requests: each request is revealed with probability  $\geq 1/2$ and,  conditioned on other requests, is revealed with probability $\geq p$.  In \Cref{sec:minimization}, we first reduce our problem from MRF distributions to the $(\frac12,p)$-sample model, and then we reduce the $(\frac12,p)$-sample problem to the $p$-sample model. This proof crucially relies on the monotonicity of the algorithm since some requests may appear in the sample with probability $\gg p$, and exploits  stochastic dominance and coupling properties of the underlying distributions.

      \paragraph{Maximization Problems.} 
    We then turn to maximization problems, particularly online allocation problems such as (hypergraph) matchings and combinatorial auctions. Here, buyers (or edges) arrive sequentially, each with a valuation drawn from a known MRF distribution, and the goal is to allocate items (or vertices) to maximize the sum of valuations (also known as \emph{social welfare}).

    A natural first step might be to  mimic the above reduction to the $p$-sample model for minimization problems. While such a reduction can be shown to exist, it fails to give $O(\Delta)$ competitive ratio:   even allocating a single item exhibits a $1/p$-hardness in the $p$-sample model (details in \Cref{apn:hardness}), limiting this approach to exponential-in-$\Delta$ guarantees. In contrast, prior work~\cite{VPS24} has shown that tight $O(\Delta)$ bounds are achievable for single-item allocation under MRFs. This highlights that the $p$-sample reduction is inherently lossy for maximization.

    To overcome this, we build on the \emph{balanced prices} framework~\cite{DuttFeldKessLuci2020}, which unifies many prophet inequality results for independent inputs. We extend this framework to handle MRF correlations and design pricing schemes that yield linear-in-$\Delta$ guarantees for general online allocation problems with MRF arrivals   
    when the buyers' valuations  are either XOS (a superclass of submodular functions) or when each buyer is interested in a bundle of at most $k$ items ($k$-hypergraph matching).  
    
   \begin{theorem}[Informal Theorem~\ref{thm:max_main}] \label{thm:informalMaxim}
        For a sequence of buyers whose valuations are drawn from a known MRF distribution with degree $\Delta$, we achieve 
        \vspace{-0.15cm}
        \begin{itemize}
            \item $O(\Delta)$-competitive algorithm for buyers with submodular/XOS valuations.
            \item $O(k^2 \cdot (\Delta + \log k) )$-competitive algorithm for $k$-uniform hypergraph matching, where each buyer is interested in a bundle of at most $k$ items.
        \end{itemize}
        \vspace{-0.15cm}
        Moreover, both these results are achieved via posted-pricing algorithms, thereby immediately enabling  online truthful mechanisms with the same approximation to the optimum welfare.
    \end{theorem}
    
    This result greatly generalizes the tight $\Theta(\Delta)$ bound of \cite{VPS24} for allocating a single item to combinatorial buyer valuations over multiple items. Even in the special cases of submodular valuations or for edge-arrival online matching, no $O(\Delta)$-competitive algorithms were previously known. 

    Balanced prices give a way to measure the expected contribution (say $b_j$) of any item $j$ to the optimum allocation, a non-trivial task beyond additive valuations. 
    The key new challenge for the MRF setting stems from the fact that rare events can potentially be correlated. 
    To overcome this, we separately handle the random ``tail'' (above $e^{4\Delta}\cdot b_j$) and ``core'' (below $e^{4\Delta}\cdot b_j$) contributions of any item $j$ using two separate pricing schemes, and then randomize between the two to obtain our final pricing. 
    
    To bound the total tail contribution, our  idea is to treat them nearly independently as they occur rarely and MRF distributions are independent up to an $e^{4\Delta}$ factor. We exploit properties of balanced prices, originally designed for product distributions, and give a generic tail bound for \emph{any} class of valuations. On the other hand, our approach to bound the total core contribution is tailored to  specific classes of valuations. For XOS valuations, this is relatively straightforward using standard properties of ``supporting prices", but for hypergraph matchings we have to carefully assign hyperedge weights to vertices.

    Taken together, \Cref{thm:full_reduction_minimization_intro} and \Cref{thm:informalMaxim} provide unified frameworks for designing online algorithms under structured correlations, addressing both minimization and maximization objectives.

    \subsection{Further Related work} \label{sec:related}
We now survey related literature across three major areas: stochastic online optimization, models of correlation, and prior work specific to MRFs.

    \paragraph{Online and prophet algorithms.} 
    There is a long line of work that gives tight bounds, when the input is worst case, for multiple online network design and coverage problems like    Steiner Tree \cite{MakoWaxm1991}, Facility Location \cite{Fota2008,Meyerson2001} and online Set Cover~\cite{AlonAA-STOC03}.  Stochastic arrival models often lead to significantly better approximations: for instance, for  independent arrivals online Steiner Tree and Facility Location admit constant-factor approximations \cite{GargGuptLeonSank2008,ArgyFrieGuptSeil2022} and Set Cover admits $\widetilde{O}(\log (mn))$-approximations \cite{GrandoniGLMSS-FOCS08,GuptKehnLevi2024}. The $p$-sample model, used in our reduction for the minimization setting, was first studied in \cite{KPSSV-ITCS19,LattanziMVWZ-NeurIPS21,CorreaCFOT21}. A similar AOS model where exactly $pn$ elements are revealed (instead of each element independently with probability $p$) was recently studied in~\cite{KNR2020,KNR-SODA22,ArgyFrieGuptSeil2022}.

    On the maximization side, $\Omega(n)$ hardness already exists for allocating a single item or, equivalently, selecting the maximum of a stream of $n$ elements. Therefore, most work has focused on stochastic arrival models.  The classic prophet inequality  \cite{KrengelS77,KrengelS78} corresponds to a $2$-approximation for our online allocation of a single item. 
    \cite{HajiaghayiKS07} pioneered combinatorial generalizations of prophet inequalities, and we now know $O(1)$-competitive algorithms for any general matroid constraint \cite{kleinberg2012matroid,FeldmanSZ16} and for  combinatorial auctions with submodular/XOS/subadditive buyers \cite{CorreaC23,FeldmanGL15}. These extensions lead to the algorithm design frameworks of  {online contention resolution schemes} \cite{FeldmanSZ16} and  {balanced pricing} \cite{kleinberg2012matroid,FeldmanGL15,DuttFeldKessLuci2020}. Interested readers are referred to book chapter \cite[Chapter 30]{EIV-Book23} and survey \cite{lucier2017economic}.

    \paragraph{Other Correlation Models.}
    Several alternative models have been proposed to capture mild or structured correlations beyond independence. In the $k$-wise independent model~\cite{CaraGravLuWang2021,GravWang2024,GuptaHKL-IPCO24,DughmiKP-STOC24}, all subsets of $k$ variables are jointly independent, which generalizes full independence while retaining analytical tractability. Linear correlation models—common in revenue optimization and market design—assume that each value is a linear combination of shared latent variables~\cite{bateni2015revenue,ImmorlicaSW23,ChawlaMS-EC10}. This has been especially useful in pricing and assortment problems, where value structure arises naturally from features. In the classic interdependent values model from Economics, each random value depends on private independent ``signals"  held by all the buyers \cite{milgrom1982theory,MMR-EC24,FMMR-EC25}.
        Other correlation frameworks include negative correlation (studied for prophet inequalities in \cite{RinoSamu1987,RinoSamu1992,ImmorlicaSW23,qiu2022submodular}), uncontentious distributions  to capture correlations in the context of contention resolution schemes \cite{Dugh2020,Dugh2025,zhao2025universal}, and correlation gaps to compare expected values under correlated versus product distributions \cite{agrawal2012price,RubiSing2017, ChekLiva2024}. These models are often specialized and do not provide a single tunable parameter interpolating between independence and full correlation—unlike MRFs.

    \paragraph{Stochastic Optimization under MRF Dependencies.}
    Markov Random Fields (MRFs) have recently emerged as a natural and general model for capturing structured correlations in stochastic optimization. They were first studied in this context by Cai and Oikonomou~\cite{CaiOiko2021}, who showed exponential-in-$\Delta$ approximations for problems such as single-buyer revenue maximization under additive or unit-demand preferences. More recently, \cite{VPS24} improved these bounds to $O(\Delta)$ for the single-item prophet inequality and revenue maximization with subadditive buyers.

    However, these results are limited to narrow problem settings and do not extend to broader online combinatorial problems such as Facility Location, Steiner Tree, or online allocations with complex buyer valuations. Our work substantially generalizes these early efforts, providing principled frameworks for both minimization and maximization problems under MRF-distributed inputs, achieving $O(\Delta)$-competitive algorithms in domains where no previous guarantees were known.

%% file: prelims.tex
\section{Preliminaries on Markov Random Fields}
In this section, we formally define Markov Random Fields, the weighted maximum degree $\Delta$, and state some of their key properties.

\begin{definition}[Markov Random Field]\label{def:MRF}
    A \emph{Markov Random Field (MRF)} $\cM$ defines a distribution $\cD_{\cM}$ over a state (\emph{type}) space $\Omega = \Omega_1 \times \dots \times \Omega_n$ and consists of a tuple 
    $\cM = (\Omega,~E,~\{\psi_i\}_{i \in [n]},~\{\psi_e\}_{e \in E})
    $, where    
    \vspace{-0.15cm}
    \begin{itemize} 
        \item $E \subseteq 2^{[n]}$ is the edge-set of a hypergraph on $[n]$.
        \item $\psi_i : \Omega_i \to \R$ is a potential function on coordinate $i$ for each $i \in [n]$.
        \item $\psi_e : \prod_{i \in e} \Omega_i \to \R$ is a potential function over hyperedge $e$ for each $e \in E$.
    \end{itemize}
    \vspace{-0.15cm}
    This distribution $\cD_{\cM}$ is then defined such that if $V  \sim \cD_{\cM}$ and $\mathbf{u} = (u_1, \dots, u_n) \in \Omega$, we have
    \[
    \Pr[V = \mathbf{u}] \propto \exp\Big(\sum_{i \in [n]} \psi_i(u_i) + \sum_{e \in E} \psi_e((u_i)_{i \in e})\Big).
    \]
\end{definition}

To build some intuition on MRFs, notice that in the absence of edges, they are equivalent to product distributions. 
The general MRF model captures several important graphical models as special cases:

\begin{itemize}
    \item Ising model: This corresponds to having a simple graph (i.e., where no edge has $3$ or more vertices) with $\Omega=\{-1,1\}^n$ and  the function $\psi_e$ taking value either $+J_e$ or $-J_e$ for some $J_e\in \Re$, depending on whether the two end points agree or disagree with each other.

    \item (Anti-)Ferromagnetic Ising model: This is a special case of Ising model with identical edge weights, i.e., $J_e = J$. It's called ferromagnetic or anti-ferromagnetic depending on the sign of $J$, and its magnitude is called the inverse temperature.
    
    \item Sherrington–Kirkpatrick model: This is another special case of Ising model where the underlying graph is complete and the edge weights $\psi_e$ are chosen from Gaussian distribution.
    
    \item Potts model: This is where we have a simple graph, and each vertex has $q$ states $\Omega_i = \{1, \dots, q\}$.
    
\end{itemize}

We define the weighted maximum degree of an MRF $\cM$ as follows.
\begin{definition}[Weighted maximum degree]
    For any MRF $\cM = (\Omega,E,\{\psi_i\}_{i \in [n]},\{\psi_e\}_{e \in E})$,  the \emph{weighted maximum degree} 
    \[
    \Delta(\cM) := \max_{i \in [n]} \max_{U \in \Omega} \Big|\sum_{e \ni i} \psi_e ((u_i)_{i \in e})\Big|.
    \]
    Moreover, suppose $\cD$ is a distribution on $\Omega$ such that $\cD = \cD_{\cM}$ for an MRF $\cM$ with $\Delta(\cM) = \Delta \geq 0$, we say $\cD$ is a \emph{$\Delta$-MRF distribution}.
\end{definition}

Note that in the definition of $\Delta(\cM)$, we only sum over $\psi_e$, and not over vertices $\psi_i$. Hence, product distributions correspond to $\Delta=0$. In the special case where the MRF corresponds to (anti-)ferromagnetic Ising model, $\Delta$ equals the inverse temperature times the max-degree of the graph.

Crucially, we will use the following property of $\Delta$-MRF distributions, which  bounds the degree to which conditioning can affect any marginal distribution.

\begin{fact}[Lemma 1 from~\cite{CaiOiko2021}] 
    \label{lem:mrf_conditioned_bound}
    Suppose the random variable $V \sim \cD$ has a $\Delta$-MRF distribution. Then for any $E_i \subseteq \Omega_i$ and $E_{-i} \subseteq \Omega_{-i}$
    \[
        \Pr[v_i \in E_i] \cdot e^{-4\Delta} \quad \leq \quad  \Pr[v_i \in E_i \mid v_{-i} \in E_{-i}] \quad  \leq \quad  \Pr[v_i \in E_i] \cdot e^{4\Delta}.
    \]
    In other words, conditioning on $v_{-i}$ changes the marginal distribution $v_i$ by at most a factor $e^{4\Delta}$ pointwise.
\end{fact}
The proof of this fact follows immediately from the definition of $\Delta$-MRF distribution.

%% file: minimization.tex
In this section, we will formally define our subadditive coverage  cost minimization problems and prove \Cref{thm:full_reduction_minimization_intro}. 
The formal definition appears in \Cref{sec:min-model}, the statement of the main theorem appears and proof outline in \Cref{sec:min-results}, and the complete proof is in \Cref{sec:min-proof}.

\subsection{Model: \probName{} Problems}\label{sec:min-model}
Similarly to \cite{GargGuptLeonSank2008}, 
let $V$ denote a ground set of possible demand points and $E$ denote the ground set of elements used to construct a solution. For any set of demand points $D \subseteq V$, we have a collection of feasible solutions denoted by $\mathcal{F}_{D}\subseteq 2^E$. To capture general coverage constraints, these solutions sets should have two key properties: (1) we must be able to extend any solution for a set of demands to a solution for a superset of the demands, and (2) the union of two feasible solutions for two different demand sets should be a feasible solution for the union of demand sets. We formalize this notion of coverage constraints bellow.

\begin{definition}[Generalized coverage constraints \cite{ArgyFrieGuptSeil2022}]\label{assu:sol_space}
Given a universe of demands $V$ and a universe of solution elements $E$, a family of \emph{generalized coverage constraints} is a collection of feasible solutions $\mathcal{F}_{D} \subseteq 2^E$ for each $D \subseteq 2^V$ such that 
\begin{enumerate}[label=(\alph*)]
        \item  For any  $D_1, D_2 \subseteq V$ and any $S_1\in \cF_{D_1}$ and $S_2\in \cF_{D_2}$, we have $S_1 \cup S_2  \in \cF_{D_1\cup D_2}$  \label{assu:space_union}
        \item  For any  $D_1\subseteq D_2\subseteq V$, we have $\cF_{D_2} \subseteq \cF_{D_1}$  \label{assu:space_subset}
    \end{enumerate}
\end{definition}

Using this notion of general coverage constraints, we can define a subadditive coverage problem as follows.

\begin{definition}[(Online) Subadditive Coverage]
    In a \emph{subadditive coverage problem}, we have a universe of demand points $V$, a universe of solution elements $E$, a collection of generalized coverage constraints $\{\mathcal{F}_{D}\}_{D\subseteq V}$, and a monotone subadditive\footnote{A function $c$ is subadditive if, for all $S_1, S_2\subseteq E$, we have $c(S_1) + c(S_2) \geq c(S_1 \cup S_2)$.} cost function $c : 2^E \to \Rp$.
    Any subset of elements $S \subseteq E$ can be selected (bought) by paying  cost $c(S)\in \mathbb{R}_{\geq 0}$.  
    Given a set of demands $D\subseteq V$, the goal is to select a subset of elements $F \in \mathcal{F}_{D}$ while minimizing cost $c(F)$.

In the online setting, a sequence of demands $D = \{v_1, \dots, v_n\}$ arrive one by one (possibly from some known arrival model).
Starting with $F_0 = \varnothing$, after each arrival $v_i$, we must select a subset of new elements $S_i \subseteq E$ such that the total set of elements selected $F_i = F_{i-1} \cup S_i$ satisfies $F_i \in \mathcal{F}_{D_{\leq i}}$, where $D_{\leq i} := \{v_1, \dots, v_i\}$ denotes the set of all current demand points. 
Our objective is to choose a feasible solution to minimize $c(F_n)$.
\end{definition}

\paragraph{Examples.}
We can capture several fundamental online minimization problems using this framework: 
\vspace{-0.15cm}
\begin{itemize}
    \item \textit{Steiner Tree}: We have a graph $G = (V, E)$ with a fixed root vertex $r \in V$. Demands are the vertices of the graph and solution elements are the edges. For a set of demands $D \subseteq V$, a feasible solution $F \in \cF_D$ is a set of edges such that all demands are connected to the root in the graph $(V, F)$. Each edge has an associated cost $c(e)$, and costs are additive.  The classic result of  \cite{MakoWaxm1991} gives a tight $\Theta(\log n)$-competitive algorithm  for online adversarial arrivals.

    \item \textit{Facility Location}: The demand points $v$ (clients) are points in a metric space $M$. Solution elements include both points $f \in M$ (facilities) and connections $(f, v) \in M\times M$ between a facility and a client. For a set of clients $D \subseteq M$, a feasible solution is a set of open facilities and connections such that each client is connected to an open facility. Costs are additive: each facility $f \in M$ has an  opening cost $c(f)$ and each connection $(f, v)$ has an assignment cost $c((f,v)) = d_M(f,v)$ given by the metric distance. We know tight $\Theta(\log n/\log\!\log n)$-competitive algorithm  for online adversarial arrivals \cite{Fota2008,Meyerson2001}.

    \item \textit{Set Cover}: Elements $E = [m]$ consist of indices of sets $S_1, \dots, S_m$, and demands consist of points $j \in \bigcup_{i} S_i$. For a set of demands $D$, a feasible solution $F \in \cF_D$ is a set on indices such that $D \subseteq \bigcup_{i \in F} S_i$. Each set $S_i$ has an associated cost $c(i)$, and costs are additive. Unless $\text{NP} \subseteq \text{BPP}$, we know  $\Theta(\log m \cdot \log n)$ is the tight competitive ratio for online adversarial arrivals \cite{AlonAA-STOC03,korman2004use}.
\end{itemize}
\vspace{-0.15cm}

\subsection{Main Result and Proof Overview} \label{sec:min-results}
This section presents our reduction for online subadditive coverage from the MRF arrival model (i.e. where $D$ is sampled from a known MRF distribution) to a model without correlations where some of the input demands are given upfront as samples.
We call this latter setting the \emph{$p$-sample} model. Using known results in this framework, our reduction would imply approximation guarantees in the MRF setting. Before presenting our main theorem, we formally define these two models.

\begin{definition}[\textbf{$\Delta$-MRF prophet model}]
    We have a  $\Delta$-MRF distribution $\cD$ over demand points $V$, which is known upfront to the algorithm. The demand set $D$ is then sampled $D \sim \cD$ and sequentially revealed to the online algorithm.
\end{definition}

\noindent In the \emph{single-sample $\Delta$-MRF prophet model}, the setup is the same as $\Delta$-MRF prophet model, except that the underlying MRF distribution is unknown, and instead the algorithm is provided a single $n$-dimensional sample from this MRF distribution.

Next, we define our $p$-sample independent model, whose definition is identical to the definition of \cite{KPSSV-ITCS19,LattanziMVWZ-NeurIPS21,CorreaCFOT21}, and is closely related to the AOS model in \cite{KNR2020,KNR-SODA22,ArgyFrieGuptSeil2022}.

\begin{definition}[\textbf{$p$-sample independent model}]\label{def:p-sample} 
There is an unknown set of possible demands $V = \{ v_1, v_2, \dots, v_n \}$. The algorithm is given a random \emph{sample} $S \subseteq V$ before the start  such that each $v_i$ belongs to $S$ independently with a known probability $p \in [0,1]$.
The remaining values $R = V \setminus S$ are the \emph{real values}, forming the demand set $D = R$ and  arriving in an adversarial order.
\end{definition}
We say that an algorithm is $\alpha(p)$-competitive in the $p$-sample independent model if its total cost is at most $\alpha(p)\cdot \opt(V)$. (We remark that an alternate choice of benchmark is $\E[\opt(R)]$: both these benchmarks are within a constant factor of each other for subadditive problems with $p<1/2$; we choose $\opt(V)$ since it makes the proofs easier as it does not depend on the randomness of the sample set $S$.)

An $\alpha(p)$-competitive algorithm in the $p$-sample independent model is called \emph{monotone} if it maintains its $\alpha(p)$-competitiveness even if the sample set is \emph{augmented} by moving some of the real values (possibly adversarially) into the sample set, i.e., the new sample set is  $S \supseteq S'$, where $S'$ contains each $v_i$ independently with probability $p$, and the demand set is $V\setminus S$. 

Our main result is a reduction from the $\Delta$-MRF prophet model to the $p$-sample independent model.

\begin{theorem}\label{thm:full_reduction_minimization}
For any $p \in [0,1]$, if  there exists an $\alpha(p)$-competitive monotone algorithm for a subadditive coverage problem in
the $p$-sample independent model, then there exists a $2\alpha\lp(\frac{1}{2}e^{-8\Delta}\rp)$-competitive algorithm for the same problem in the  $\Delta$-MRF model.
Specifically, an $O\lp( \log ({1}/{p})\rp)$-approximation in the $p$-sample model yields an $O(\Delta)$-approximation in the MRF model.
{Moreover, this reduction holds even from the single-sample $\Delta$-MRF prophet model to the $p$-sample independent model.}
    \end{theorem}

As mentioned in \Cref{sec:intro-results}, the main application of \Cref{thm:full_reduction_minimization} is the implication of $O(\Delta)$-competitive algorithms for online subadditive coverage problems with $O(\log(1/p))$-competitive monotone algorithm  in the $p$-sample independent model. In particular, in \Cref{sec:monotoneAlgos} we show that existing algorithms (or their variants) due to \cite{GargGuptLeonSank2008,ArgyFrieGuptSeil2022} satisfy this property for online facility location and online Steiner tree, which proves \Cref{cor:log-guarantee-min}.

\paragraph{Proof  Overview.}

In order to obtain the reduction in \Cref{thm:full_reduction_minimization}, we define two intermediate problems: the \emph{MRF Game of Googol} (\Cref{def:GoG}) and the \emph{$(\frac{1}{2},p)$-sample} model (\Cref{def:p-q-sample}). For the Game of Googol, similarly to the original definition of~\cite{CorrChriEpstSoto2020}, we want for each value to randomize between two possible outcomes, but using an MRF distribution instead of a uniformly random choice. For the $(\frac12,p)$-sample model, we require each demand to be in the demand set with some minimum probability.  

\begin{definition} [\textbf{$\Delta$-MRF Game of Googol}] \label{def:GoG}
An adversary chooses arbitrary values $v^1_1, \dots, v^1_n, v^{-1}_1, \dots, v^{-1}_n$ unknown to the algorithm. We then sample $\bm{\sigma} \in \{-1, 1\}^n$ from an \emph{unknown} symmetric\footnote{Symmetric  means that for any $\bm{\sigma} \in \{-1, 1\}^n$, both $\bm{\sigma}$ and $-\bm{\sigma}$ have the same probability.} $\Delta$-MRF distribution $\cD_\sigma$ over $\{-1, 1\}^n$. When a sample $\bm{\sigma} \sim \cD_\sigma$ is drawn, the values $\bm{v}^{\bm{\sigma}} = (v^{\sigma_1}_1, \dots, v^{\sigma_n}_n)$ become the sample set $S = \bm{v}^{\bm{\sigma}}$ and are revealed to the algorithm offline, while the values $\bm{v}^{-\bm{\sigma}} = (v^{-\sigma_1}_1, \dots, v^{-\sigma_n}_n)$  are the real values $R = \bm{v}^{-\bm{\sigma}}$ and arrive sequentially as demands $D = R$. (Note that symmetry implies  $S$ and $R$ are identically distributed.)
\end{definition}

\begin{definition}[\textbf{$(\frac12,p)$-sample MRF model}]\label{def:p-q-sample}
There is an \emph{unknown} adversarial value set $V = \{ v_1, v_2, \dots, v_n \}$. A random sample $S = (s_1, s_2, \dots, s_k) \subseteq V$ is revealed to the algorithm upfront such that the random vector $(\ind{v_i \in S})_{i \in [n]}$ has a $\Delta$-MRF distribution {over $\{0,1\}^n$}. Additionally, we assume that this \emph{unknown} MRF distribution satisfies the following  for each $v_i \in V$:
\begin{itemize}
    \item $\Pr[v_i \in S] \geq 1/2$,
    \item $\Pr[v_i \in S \mid v_{-i}] \geq p$ for any conditioning of other values $v_{-i}:=V\setminus \{i\}$.
\end{itemize}
The remaining values $R := V \setminus S$ form the demand set $D = R$ and arrive sequentially in an adversarial order. 
\end{definition}

We prove Theorem~\ref{thm:full_reduction_minimization} by sequentially reducing from the initial single-sample MRF prophet problem to the Game of Googol MRF, to the $(\frac{1}{2},p)$-sample, and finally to the $p$-sample independent model. Figure~\ref{fig:tikz_reductions_min} shows this reduction, and the main lemmas are given in Sections~\ref{subsec:MRF_to_googol}, \ref{subsec:googol_to_half_p}, and \ref{subsec:half_p_to_p}.

\begin{figure}[H]
    \centering
    \input{tikz_min_reductions}
    \caption{Outline of the sequence of reductions for the minimization case. An arrow $A \rightarrow B$ means problem A reduces to problem B.}
    \label{fig:tikz_reductions_min}
\end{figure}
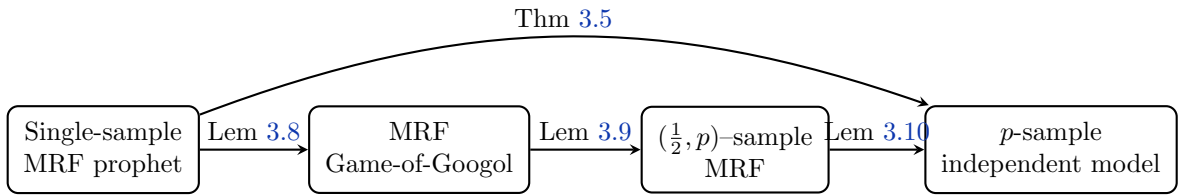

    \begin{proof}[Proof of Theorem~\ref{thm:full_reduction_minimization}]
        Starting with an instance of (single-sample) $\Delta$-MRF prophet, we sequentially convert it first through \Cref{lem:prophet-to-googol} to an instance of a $2\Delta$-MRF Game of Googol, then through \Cref{lem:half_p_to_game_of_googol} to an instance of $(\frac{1}{2},p)$-sample MRF with $p=\frac{1}{2}e^{-8\Delta}$, and finally using \Cref{lem:p-sample_to_half_p_mrf} to a $p$-sample independent instance, with $p=\frac{1}{2}e^{-8\Delta}$. Therefore, an $\alpha(p)$-approximation for the $p$-sample model implies a $2\alpha\lp(\frac{1}{2}e^{-8\Delta}\rp)$-approximation for the initial MRF minimization model.
    \end{proof}

\subsection{Proof  of \Cref{thm:full_reduction_minimization}}\label{sec:min-proof}

We now prove our three reductions in \Cref{lem:prophet-to-googol}, \Cref{lem:half_p_to_game_of_googol}, and  \Cref{lem:p-sample_to_half_p_mrf}.

\subsubsection{Reducing Single-Sample $\Delta$-MRF Prophet  to Game of Googol MRF}\label{subsec:MRF_to_googol}

In this section, we show that any online optimization problem in the single-sample $\Delta$-MRF prophet model can be reduced to the Game of Googol MRF. 

\begin{lemma}\label{lem:prophet-to-googol}
    The single-sample $\Delta$-MRF prophet model can be reduced to a random $(2\Delta)$-MRF Game of Googol instance.
\end{lemma}
\begin{proof}
    Let $\cD$ be a $\Delta$-MRF distribution, and consider a single-sample prophet setting where $S \sim \cD$ and $R = D \sim \cD$ are the sets of samples and real values respectively. Let $\bm{\sigma}\in \{-1,1\}^n$ be a uniformly random vector, drawn independently of $S$ and $R$. Finally, let $V = (v^1_1, \dots, v^1_n, v^{-1}_1, \dots, v^{-1}_n)$ be given by
    \[
    \forall i \in [n],\quad 
    (v_i^1, v^{-1}_i) = \begin{cases}
        (s_i, r_i) & \text{if } \sigma_i=1, \\
        (r_i, s_i) & \text{if } \sigma_i=-1 .
    \end{cases}
    \]
    We show that for any realization of $V$, the distribution $\cD^V_{\bm \sigma}$ of $\bm \sigma$ conditional on $V$ can be represented as a $2\Delta$-MRF. (Note that this distribution depends on $V$, so is unknown to the algorithm, but we only need its existence.)
    Let $\psi_i$ for $i \in [n]$ and $\psi_e$ for $e \in E$ be the vertex and edge potentials in the MRF $\cD$ respectively. From Definition~\ref{def:MRF}, we know that for any fixed $v_1, \dots, v_n$, 
    \[
    \Pr_{S \sim \cD}[S = (v_1, \dots, v_n)] \propto \exp\Big(\sum_{i \in [n]} \psi_i(v_i) + \sum_{e \in E}\psi_e(v_e)\Big).
    \]
    Using this, we compute the probability mass function of $\cD^V_{\bm \sigma}$. The idea is that $V = V(S, R, \bm \sigma')$ depends on three independent sources of randomness: the MRF samples $S \sim \cD$ and $R \sim \cD$, as well as the uniform $\bm \sigma' \sim \mathrm{Unif}\{-1,1\}^n$. We show that, if given only the result $V$, then the conditional distribution on $\bm \sigma'$ can be expressed in terms of the distributions of $S$ and $R$ using Bayes' theorem:
\begin{align*}
    \Pr_{\bm \sigma', S, R}[\bm \sigma' = \bm \sigma \mid V(S, R, \bm \sigma') = V] &= 
    \frac{\Pr_{\bm \sigma',S, R}[V(S, R, \bm \sigma') = V \mid \bm \sigma' = \bm \sigma] \cdot \Pr_{\bm \sigma'}[\bm \sigma' = \bm \sigma]}{\Pr_{\bm \sigma', S, R}[V(S, R, \bm \sigma') = V]} \\
 &   = 
    \frac{\Pr_{\bm \sigma',S, R}[V(S, R, \bm \sigma') = V \mid \bm \sigma' = \bm \sigma] \cdot 2^{-n} }{\Pr_{\bm \sigma', S, R}[V(S, R, \bm \sigma') = V]} \\
    &\propto \Pr_{\bm \sigma',S, R}[V(S, R, \bm \sigma') = V \mid \bm \sigma' = \bm \sigma] \\
    &= \Pr_{S, R}[V(S, R, \bm \sigma) = V] \\
    &= \Pr_{S}[S= v^{\bm \sigma}]\cdot \Pr_{R}[R= v^{-\bm{\sigma}}]
\end{align*}
    
    Next, let $\bm{v}^{\bm{\sigma}} := (v_1^{\sigma_1}, \dots, v_n^{\sigma_n})$, and likewise define $\bm{v}^{-\bm{\sigma}}$. We have
    \begin{align*}
        \Pr_{\bm{\sigma'}\sim \cD^V_{\bm \sigma}}[\bm{\sigma'} = \bm \sigma] &\propto \Pr_{S \sim \cD}[S = \bm{v}^{\bm{\sigma}}] \cdot \Pr_{R \sim \cD}[R = \bm{v}^{-\bm{\sigma}}]\\
        &\propto \exp\Big(\sum_{i \in [n]} (\psi_i(v^{\sigma}_i) + \psi_i(v^{-\sigma}_i)) + \sum_{e \in E}(\psi_e(v^{\sigma}_e) + \psi_e(v^{-\sigma}_e))\Big)\\
        &\propto \exp\Big(\sum_{e \in E}(\psi_e(v^{\sigma}_e) + \psi_e(v^{-\sigma}_e))\Big).
    \end{align*}
    Therefore, $\cD^V_{\bm \sigma}$ can be represented as an MRF with potential functions $\hat\psi_i(\sigma_i) = 0$ and $\hat \psi_e(\bm \sigma_e) = \psi_e(v^{\sigma}_e) + \psi_e(v^{-\sigma}_e)$, and observe that by our construction the maximum degree of this MRF is at most $2\Delta$. Additionally, it is not hard to see that $\cD_{\bm{\sigma}}^V$ is symmetric, has uniform marginals, and is known to the algorithm as it is determined by $V$ and $\cD$. Hence, we this gives a valid game of googol instance.
      
  Completing the proof, the original single-sample prophet instance can be generated by first sampling $V$ and $\bm{\sigma} \sim \cD^V_{\bm \sigma}$, and then setting $S = \bm{v}^{\bm{\sigma}}$ and $R = \bm{v}^{-\bm{\sigma}}$. With these sets, we can run $\alpha$-competitive algorithm for the $(2\Delta)$-MRF Game of Google model to get an $\alpha$-competitive ratio in the  $\Delta$-MRF setting.
\end{proof}

\subsubsection{Reducing Game of Googol to \texorpdfstring{$(\frac{1}{2}, p)$-sample}{(1/2, p)-sample} MRF}\label{subsec:googol_to_half_p}

Now we reduce the Game of Googol model to the $(\frac{1}{2},p)$-sample MRF. The reduction works by splitting the initial Game of Googol instance into two disjoint instances for the $(\frac{1}{2}, p)$-MRF problem, then using the approximation for this problem we are able to map the decisions one to one in our initial instance. We formally show the following.

\begin{lemma}\label{lem:half_p_to_game_of_googol}
For $p = \frac{1}{2}\exp(-4\Delta)$, if there is an $\alpha$-approximation algorithm for the $(\frac{1}{2}, p)$-sample MRF model then there is a $2\alpha$-approximation for the Game of Googol model with max degree $\Delta$.
\end{lemma}
\begin{proof}[Proof of Lemma 3.9]
    Our algorithm will partition the GoG instance with sample/real sets $(S, R)$ into two subinstances $(S_1, R_1)$ and $(S_2, R_2)$, and run the $(\frac{1}{2}, p)$-sample algorithm on each subinstance. To create this partition, we will flip independent coins for each $i$. If coin $i$ is heads, we add $s_i$ to $S_1$ and $r_i$ to $R_2$. If tails, we add $s_i$ to $S_2$ and $r_i$ to $R_1$. This ensures $|S_1 \cup R_1| = |S_2 \cup R_2| =n$, so each subinstance is a valid input to the $(\frac{1}{2},p)$-sample algorithm. It now only remains to bound the performance of this algorithm.

    Formally, let $\mathcal{I}$ be an instance of the Game of Googol model, and denote by $ \opt_{GoG}(\mathcal{I})$ and $\alg_{GoG}(\mathcal{I})$ our final algorithm's and the optimal solution cost for $\cI$ respectively.
    Denote by $v^1_1, \dots, v^1_n$ and $v^{-1}_1, \dots, v^{-1}_n$ the values in $\cI$, and by $\bm{\sigma} \sim \cD_{\bm{\sigma}}$ the MRF where $\bm{\sigma} \in \{-1, 1\}^n$. 
    This defines a sample set $S =  (v^{\sigma_1}_1, \dots, v^{\sigma_n}_n)$ and a real set $R = (v^{-\sigma_1}_1, \dots, v^{-\sigma_n}_n)$ in the instance $\mathcal{I}$.
    
    To reduce to the $(\frac{1}{2}, p)$-sample MRF setting, we will split the instance $(S, R)$ into two subinstances $(S_1, R_1)$ and $(S_2, R_2)$ as follows. First, we sample $\bm{\Tilde \sigma} \sim \mathrm{Unif}\{-1, 1\}^n$. Then, we define
    \begin{align*}
        S_1 &:= \{ s_i \in S : \Tilde \sigma_i = 1 \} \\
        R_1 &:= \{ r_i \in R : \Tilde \sigma_i = -1 \} \\
        S_2 &:= \{ s_i \in S : \Tilde \sigma_i = -1 \} \\
        R_2 &:= \{ r_i \in R : \Tilde \sigma_i = 1 \}
    \end{align*}
    Using these sets, we construct two instances: \( \cI_1 = (S_1, R_1) \) and \( \cI_2 = (S_2, R_2) \), where for $i\in\{1,2\}$, we have \( S_i \) represent the sample set and \( R_i \) represent the real set.

    Starting with $\cI_1$, we seek to show that the instance $(S_1, R_1)$ has a valid $(\frac{1}{2}, p)$-sample MRF distribution conditional on any realization of $S_1 \cup R_1 = \{v_i^{\Tilde \sigma_i \sigma_i} : i \in [n] \}$. Notice that for any value $v_i \in S_1 \cup R_1$, we have: 
    \[
    \Pr[v_i \in S_1 \mid S_1 \cup R_1] = \Pr[ \sigma_i = 1 \mid (\Tilde \sigma_j \sigma_j)_{j \in [n]}] = \Pr[\sigma_i = 1] = \frac{1}{2}
    \]
    by the definition of Game of Googol MRF and the independence of $(\Tilde \sigma_j \sigma_j)_{j \in [n]}$ and $\bm \sigma$, and
    \[
    \Pr[v_i \in S_1 \mid S_1 \setminus \{v_i\},~S_1 \cup R_1] = \Pr[\sigma_i = 1 \mid \sigma_{-i},~(\Tilde \sigma_j \sigma_j)_{j \in [n]}] = \Pr[\sigma_i = 1 \mid \sigma_{-i}] \geq \frac{1}{2} \exp\left( -4 \Delta \right)
    \]
    by Fact 2.3 and the aforementioned independence.

    The equivalent property also holds for $\cI_2$, so both $\cI_1$ and $\cI_2$ are (random) instances for the $(\frac{1}{2}, p)$-sample MRF for $p = \frac{1}{2} \exp\left( -4 \Delta \right)$. Denote by \( \alg_{MRF} \) the cost of the \(\alpha(p)\)-approximation algorithm for the $(\frac{1}{2}, p)$-sample MRF model (we omit the $p$ for brevity) and denote by $\opt_{MRF}$ the optimal cost.
    
    For the instances $\cI_1$ and $\cI_2$ we run $\alg_{MRF}$. In each instance the algorithm buys a set of elements $F_{R_1}$ (resp. $F_{R_2}$) to satisfy the incoming demands from $R_1$ (resp. $R_2$).
    In the initial instance we sequentially see demands from both $\cI_1$ and $\cI_2$ mixed with each other. Upon seeing demand $v_i\in R$, we add to the solution set the elements $e_i \in \cF_{R_1}$ given by $\alg_{MRF}$ to satisfy $v_i \in R_1$ (resp. $e_i\in \cF_{R_2}$ if $v_i\in R_2$). 
    Therefore the final solution of $\alg_{GoG}$ is $\cF_{R_1}\cup \cF_{R_2}$ which is a feasible solution for the Game of Googol MRF by Assumption 3.1, (a). The final costs of our reduction are
    \begin{align*}
        \alg_{GoG}(\mathcal{I}) & \leq\alg_{MRF}(\cI_1) + \alg_{MRF}(\cI_2)  & \text{Subadditivity}\\
        & \leq  \alpha \opt_{MRF}(\cI_1) +  \alpha \opt_{MRF}(\cI_
2) & \alpha\text{-approx guarantee}\\
        & \leq \alpha \opt_{MRF}(\cI_1 \cup \cI_2) + \alpha \opt_{MRF}(\cI_1 \cup \cI_2) & \text{Monotonicity, (Assumption 3.1, (b))}\\
        & = 2\alpha \opt_{GoG}(\mathcal{I}) & \text{Definition of }\cI_1, \cI_2.    
    \end{align*}
    \end{proof}

\subsubsection{Reducing \texorpdfstring{$(\frac{1}{2}, p)$-sample}{(1/2, p)-sample} MRF to \texorpdfstring{$p$-sample}{p-sample} independent}\label{subsec:half_p_to_p}

Finally, we observe that any approximation guarantee in the $p$-sample model immediately translates to the $(\frac{1}{2}, p)$-sample model.

\begin{lemma}\label{lem:p-sample_to_half_p_mrf}
If there exists an $\alpha(p)$-competitive monotone algorithm in the $p$-sample independent model, then there exists an $\alpha(p)$-competitive for the $(\frac{1}{2}, p)$-sample MRF model.
\end{lemma}
\begin{proof}
For any instance, our algorithm in the $(\frac{1}{2}, p)$-sample MRF model runs the $\alpha(p)$-approximation monotone subroutine in the $p$-sample independent model. 
This is possible since we can send the sample set in $(\frac{1}{2}, p)$-sample model as input to the $p$-sample independent subroutine. We claim that this algorithm is  $\alpha(p)$-competitive in the $(\frac{1}{2}, p)$-sample MRF model.

To prove this, we show that the sample-set sent to the $p$-sample  subroutine can be viewed as first drawing each value of the adversarial set $V=\{v_1,\ldots,v_n\}$ independently with probability $p$ and then augmenting it with some additional real values. Due to monotonicity of the algorithm, the algorithm remains $\alpha(p)$-competitive even after this augmentation.

To this end, we observe that the sample-set distribution $\mathcal{D}_1$ in the $(\frac{1}{2}, p)$-sample majorizes the sample-set distribution  $\mathcal{D}_2$ in the $p$-sample independent model. In other words, there exists a coupling between $\mathcal{D}_1$ and $\mathcal{D}_2$ such that $S_1 \supseteq S_2$ with   $S_1 \sim \mathcal{D}_1 $ and  $S_2 \sim \mathcal{D}_2$.  Hence, $S_1\setminus S_2$ can be viewed as an augmentation. The coupling exists because we can sample the elements of $S_1$ one by one, going from $v_1$ to $v_n$, from the conditional distribution (i.e., sample whether $v_i$ is in $S_1$ after conditioning on $v_1,\ldots, v_{i-1}$);  each value $v_i$ appears in $S_1$  with probability at least $p$ in this conditional distribution, irrespective of the prior conditioning, due to \Cref{def:p-q-sample}.  Hence, $\mathcal{D}_1$ majorizes  $\mathcal{D}_2$.
\end{proof}

%% file: tikz_min_reductions.tex
\begin{tikzpicture}[node distance=4.2cm,  >=stealth,  thick]
  \tikzstyle{box} = [rectangle, draw, rounded corners, align=center, inner sep=6pt]
  \node[box] (A) {Single-sample\\ MRF prophet};
  \node[box,right of=A] (B) {MRF\\ Game-of-Googol};
  \node[box,right of=B] (C) {$(\tfrac12,p)$–sample\\ MRF};
  \node[box,right of=C] (D) {$p$-sample\\ independent model};
  
  \draw[->] (A) -- node[above]{Lem \ref{lem:prophet-to-googol}} (B);
  \draw[->] (B) -- node[above]{Lem~\ref{lem:half_p_to_game_of_googol}} (C);
  \draw[->] (C) -- node[above]{Lem~\ref{lem:p-sample_to_half_p_mrf}} (D);

  \draw[->, bend left=20] (A) to node[above]{Thm~\ref{thm:full_reduction_minimization}} (D);

\end{tikzpicture}

%% file: maximization.tex
Now we present our results for  online allocation problems such as online combinatorial auctions and $k$-uniform hypergraph matching. 

In an \emph{online allocation problem},  we have $m$ items and $n$ buyers, where the buyers arrive one by one with their valuations 
$V = (v_1, \dots, v_n)$ 
 drawn from some known correlated distribution $\mathcal{D}$, where $v_i : 2^{[m]}\rightarrow \R_{\geq 0}$. The goal is to immediately allocate the $i$-th buyer a subset of the remaining items to maximize the total expected welfare $\E_V \lp[ \sum_{i\in[n]} v_i(S_i)\rp]$. We will assume that $\mathcal{D}$ is an MRF distribution over $\Omega = \Omega_1\times \ldots \times \Omega_n$, where any element of $\Omega_i$ is called a \emph{type} of buyer $i$ and each type corresponds to a unique valuation function.

We study two cases of valuation functions that lead to the following two settings: (1) online combinatorial auctions with XOS valuations, and (2) $k$-uniform hypergraph matching. For both cases we design \emph{posted-price} algorithms---each incoming buyer faces fixed prices $p_1,\ldots, p_m \in \R_{\geq 0}$ and takes the  {utility} maximizing decision  $S_i = \arg\!\max_{S  \in \text{Remaining}} \big( v_i(S) - \sum_{j \in S} p_j \big)$---that guarantee linear in $\Delta$ fraction of the optimal welfare. Note that our results imply $O(\Delta)$-competitive algorithm for online matching in general graphs since this is captured by $k$-uniform hypergraph matching for $k=2$.

\begin{theorem}\label{thm:max_main}
Let $\cD$ be an $\Delta$-MRF distribution over valuation profiles $V = (v_1, \dots, v_n)$. Then, we have posted-pricing algorithms with the following guarantees:
\begin{itemize}
    \item \textbf{XOS valuations}. When each buyer's valuation has the form $v_i(S) = \max_{a \in \cA_i} \sum_{j \in S} a_j$, where each $\cA_i$ is a subset of $\Rp^m$, then there is an $O({\Delta})$-competitive algorithm. 
    
    \item \textbf{$k$-uniform Hypergraph matching}. When each buyer's valuation has the form $ v_i(S) = w_i \cdot \ind{e_i \subseteq S}$ 
for some $w_i \geq 0$ and hyper-edge $e_i \subseteq [m]$ with $|e_i| \leq k$, then there is an $O({ k^2 \cdot (\Delta + \log k)})$-competitive algorithm.
\end{itemize}    
\end{theorem}

To design the posted-pricing in the proof of   \Cref{thm:max_main},   we use two key tools: the framework of ``balanced prices'' from \cite{DuttFeldKessLuci2020} and a core-tail decomposition of item values.

Roughly, {balanced prices} give a way to measure the value of an item within a fixed valuation profile $V$.  For a given realization of the valuations $V$, let $O^V_i$ denote the set of items that buyer $i$ receives in the hindsight optimum allocation with valuations $V$, then balanced prices are formally defined as follows, which is a slight modification from \cite{DuttFeldKessLuci2020} to fit our setting.

\begin{definition}[Balanced Prices]\label{def:balanced-prices}
    For a given valuation profile $V = (v_1, \dots, v_n)$ and $\alpha, \beta \geq 1$, we say that prices $p^V_1, \dots, p^V_m$ on items are \emph{$(\alpha,\beta)$-balanced} if for any buyer $i$ and set $S \subseteq [m]$, we have
    \begin{enumerate}
        \item $\sum_{j \in O^V_i \setminus S} p^V_j \geq \frac{1}{\alpha}\left(v_i(O^V_i) - v_i(O^V_i \cap S)\right)$.
        \item $\sum_{j \in O^V_i}p^V_j \leq \beta v_i(O^V_i)$.
    \end{enumerate}
\end{definition}

We need to use the balanced pricing scheme $p^V_j$ to set prices $p_j$ that are simultaneously not too low, in which case the items are bought by low value buyers first, and not too high, in which case they end up not bought by even the high value buyers. A unique challenge that comes from the correlated setting is that multiple buyers might rarely and simultaneously value the same item $j$ at an uncommonly high value, making it difficult for the item to be ``reserved'' for the highest bidder. 

To handle this challenge, we decompose the value $p^V_j$ of any item $j$ into a ``tail'' part $(p^V_j - e^{4\Delta} b_j)^+$ and a ``core'' part $\min\{p^V_j, e^{4\Delta} b_j\}$, where for any item $j$ the we define the \emph{base price} $b_j$ as
\[ b_j := \E_V \lp[ p^V_j\rp] .\]  
Now we will design two pricing schemes, which will separately  obtain welfare  approximating the contribution from the tail and core, respectively. The following lemma shows that such a guarantee would suffice since we can then randomize over the two pricing schemes.

\begin{lemma}\label{cor:final_guarantee_max}
    For some distribution $V \sim \cD$ and balanced prices $p^V_1, \dots, p^V_m$, suppose there exist two posted-pricing schemes obtaining expected welfare $\alg_1$ and $\alg_2$, respectively, such that
    \begin{align}
        \alg_1 &\geq \alpha \cdot \E_V \Big[ \sum_{j\in [m]} (p^V_j - e^{4\Delta} b_j)^+ \Big] + \Big(\opt - \alpha \sum_{j\in [m]} b_j\Big) &\text{(Tail contribution)} \label{eq:tail-price-req} \\
        \alg_2 &\geq \frac{1}{\gamma} \Big(\E_V \Big[ \sum_{j\in [m]} \min\{p^V_j,~e^{4\Delta}b_j\}\Big] - \epsilon \sum_{j\in [m]} b_j \Big)  &\text{(Core contribution)}  \label{eq:core-price-req}
    \end{align}
    for some $\gamma, \epsilon \geq 0$. {Then}, there exists a $\big(\frac{1 + \alpha \gamma}{1 - \epsilon \alpha \beta}\big)$-competitive posted-price mechanism.
\end{lemma}
\begin{proof}[Proof of \Cref{cor:final_guarantee_max}]
    Choose item prices as follows: with probability $\frac{1}{1 + \alpha \gamma}$, set prices according to $\alg_1$, and with the remaining probability, set prices according to $\alg_2$. Then, the expected welfare of the combined algorithm is
    \begin{align*}
        \E[\alg] &\geq \frac{1}{1 + \alpha \gamma}\Bigg(\alpha \E_V \Big[\sum_{j\in [m]} (p^V_j - e^{4\Delta}b_j)^+\Big] + \opt - \alpha \sum_{j\in [m]} b_j\Bigg) \\
        & \hspace{2cm}+ \frac{\alpha \gamma }{1 + \alpha \gamma} \cdot\frac{1}{\gamma} \Bigg( \E_V \Big[\sum_{j\in [m]} \min\{p^V_j,~e^{4\Delta}b_j\}\Big] - \epsilon \sum_{j\in [m]} b_j \Bigg)\\
        &=\frac{1}{1 + \alpha \gamma}\Bigg(   \alpha \E_V \Big[\sum_{j\in [m]} p_j^V\Big] + \opt - (\alpha + \alpha \epsilon) \sum_{j\in [m]} b_j\Bigg)\\
        &=\frac{1}{1 + \alpha \gamma}\Big(\opt -  \alpha\epsilon \sum_{j\in [m]} b_j\Big) \quad \geq \quad  \frac{1 - \alpha\epsilon \beta}{1 + \alpha \gamma}\opt,
    \end{align*}
    where in the last equality we used property 2 of \Cref{def:balanced-prices} (Balanced Prices).
\end{proof}

Thus, it only remains to show that we can obtain pricing schemes that satisfy the conditions of \Cref{cor:final_guarantee_max}, which we do in the following lemmas.

\begin{restatable}[Tail contribution]{lemma}{TailContribution}\label{lem:balanced-prices-tail}
    Let $V$ be drawn from a $\Delta$-MRF distribution $\cD$. Suppose for each $V$ in the support of $\cD$, we have $(\alpha,\beta)$-balanced prices $p^V_1, \dots, p^V_m$.  Suppose we run a posted-priced algorithm with item prices $p_j := \alpha e^{4\Delta} b_j$ for $j\in [m]$. 
    If $S_i$ is the set of items that buyer $i$ receives, then 
    \[
    \E_V \Big[ \sum_{i\in [n]} v_i(S_i) \Big] \geq \alpha \cdot \E_V \Big[ \sum_{j\in [m]} (p^V_j - e^{4\Delta} b_j)^+ \Big] + \Big(\opt - \alpha \sum_{j\in [m]} b_j\Big).
    \]
\end{restatable}

\begin{lemma}[Core contribution]\label{lem:balanced-prices-core}
    Let $V$ be drawn from a $\Delta$-MRF distribution $\cD$. 
    \begin{enumerate}
        \item If $\cD$ is a distribution over XOS valuation profiles, then there exist $(1,1)$-balanced prices $p^V_1, \dots, p^V_n$ and a pricing scheme which obtains expected welfare
      \[
    \frac{1}{O(\Delta)} \Bigg(\E_V \Big[ \sum_{j\in [m]} \min\{p^V_j,~e^{4\Delta}b_j\}\Big] - \frac{1}{e} \sum_{j\in [m]} b_j\Bigg).
    \] 
    
    \item If $\cD$ is a distribution over $k$-uniform hypergraph matching valuations, then there exist $(1,k)$-balanced prices $p^V_1, \dots, p^V_n$ and a pricing scheme which obtains expected welfare
      \[
    \frac{1}{O(k^2(\Delta + \log k))} \Bigg(\E_V \Big[ \sum_{j\in [m]} \min\{p^V_j,~e^{4\Delta}b_j\}\Big] - \frac{1}{ek} \sum_{j\in [m]} b_j\Bigg).
    \] 
    \end{enumerate}
\end{lemma}

We are now ready to show our main theorem. 

\begin{proof}[Proof of \Cref{thm:max_main}]
    The theorem follows immediately when we combine the core and tail pricing schemes from \Cref{lem:balanced-prices-tail} and \Cref{lem:balanced-prices-core} with \Cref{cor:final_guarantee_max}.

    First, notice that \Cref{lem:balanced-prices-tail} ensures we can always satisfy \eqref{eq:tail-price-req} for any distribution $\cD$ that admits $(\alpha, \beta)$-balanced prices. For XOS valuations, \Cref{lem:balanced-prices-core} gives us $(\alpha,\beta) = (1,1)$-balanced prices that also have a pricing schemes satisfying requirement \eqref{eq:core-price-req} of \Cref{cor:final_guarantee_max} with $\gamma=O(\Delta)$ and $ \epsilon = 1/e$. Similarly, for the hypergraph matching case, \Cref{lem:balanced-prices-core} gives us $(\alpha,\beta) = (1,k)$-balanced prices that also have a pricing scheme satisfying \eqref{eq:core-price-req} with $\gamma=O(k^2(\Delta+\log k))$ and $\epsilon = \frac{1}{ek}$ which implies the $O((\Delta + \log k) k^2)$-approximation of the theorem.
\end{proof}

\subsection{Tail Contribution}
Now we prove \Cref{lem:balanced-prices-tail}, which crucially relies on the properties of balanced prices (\Cref{def:balanced-prices}).


\begin{proof}[Proof of \Cref{lem:balanced-prices-tail}]
  
    We separately bound the utility $\sum_i(v_i(S_i) - \sum_{j \in S_i} p_j)$ and the revenue $\sum_{j \in [m]} \ind{i \text{ sold}} \cdot p_j$ and combine them to get the lemma. 
    \paragraph{Bounding the utility} Let $R^V_i$ be the set of items remaining when buyer $i$ arrives (as a function of $V$), and let $\hat V = (\hat v_1, \dots, \hat v_n)$ be a new sample from $\cD$ conditional on $\hat v_i = v_i$. Recall that $O^V_i$ denotes the set of items that buyer $i$ receives in the hindsight optimum allocation with valuations $V$.
    We consider the following alternative set of items $S_i'$ that buyer $i$ could have bought defined as
    \[
    S'_i := \arg\!\max_{S \subseteq O^{\hat V}_i \cap R^V_i} \Big[v_i(S) - \sum_{j \in S} p_j\Big].
    \]
    Since $S'_i$ is a feasible set of items when buyer~$i$ arrives, and our algorithm takes the utility-maximizing decision $S_i$,
    buyer~$i$ must receive at least as much utility from $S_i$ as from $S'_i$, so we have
    \[
    \E_V\Big[v_i(S_i) - \sum_{j \in S_i} p_j\Big] 
    ~\geq~ 
    \E_{V,\hat V} \Big[v_i(S'_i) - \sum_{j \in S'_i} p_j ~\Big | ~ \hat v_i = v_i \Big] 
    ~=~
    \E_{V,\hat V} \Big[ \max_{S \subseteq O^{\hat V}_i \cap R^V_i} v_i(S) - \sum_{j \in S} p_j ~\Big | ~ \hat v_i = v_i \Big].
    \]
    Next, by  property 1 of \Cref{def:balanced-prices}, we have $v_i(S) \geq v_i(O^{\hat V}_i \cap S) \geq  v_i(O^{\hat V}_i) - \alpha \sum_{j \in O^{\hat V}_i \setminus S}p^{\hat V}_j$. This gives
    \begin{align}
        \E\Big[v_i(S_i) - \sum_{j \in S_i} p_j\Big] 
        &\geq \E_{V,\hat V} \Bigg[ \max_{S \subseteq O^{\hat V}_i \cap R^V_i} v_i(O^{\hat V}_i) - \alpha \sum_{j \in O^{\hat V}_i \setminus S}p^{\hat V}_j - \sum_{j \in S} p_j ~\Big | ~ \hat v_i = v_i \Bigg] \nonumber\\
        &\geq \E_{V,\hat V} \Bigg[\max_{S \subseteq O^{\hat V}_i \cap R^V_i}v_i(O^{\hat V}_i) - \alpha \sum_{j \in O^{\hat V}_i}p^{\hat V}_j + \sum_{j \in S}(\alpha p^{\hat V}_j - p_j) ~\Big | ~ \hat v_i = v_i \Bigg] .\nonumber
    \end{align}
    Since the value of $v_i(O^{\hat V_i}) - \alpha \sum_{j \in O^{\hat V}_i}p^{\hat V}_j$ is independent of the choice of $S$, the expression is maximized when $\sum_{j \in S}(\alpha p^{\hat V}_j - p_j)$ is maximized. Therefore, $S = \{j \in O^{\hat V}_i \cap R^V_i : \alpha p_j^{\hat V} \geq p_j\}$.
    Recalling that $p_j = \alpha e^{4\Delta} b_j$, we have
    \begin{equation}
        \E\Big[v_i(S_i) - \sum_{j \in S_i} p_j\Big] \geq \E_{V,\hat V} \Bigg[v_i(O^{\hat V}_i) - \alpha \sum_{j \in O^{\hat V}_i}p^{\hat V}_j + \alpha\sum_{j \in O^{\hat V}_i \cap R^V_i}(p^{\hat V}_j - e^{4\Delta}b_j)^+ ~\Big | ~ \hat v_i = v_i \Bigg].
        \label{eq:tail-util}
    \end{equation}

    \paragraph{Bounding the revenue} Similarly to the utility proof, we imagine a sample $V' \sim \cD$ independent from both $V$ and $\hat{V}$. We have
    \begin{align*}
        \E_V\Big[\sum_{j\in [m]} \ind{i \text{ sold}} \cdot p_j\Big]
        ~=~ \E_V\Bigg[\sum_{j \in [m] \setminus R^V_n}  \alpha e^{4\Delta} \E_{V'} \lp[p^{V'}_j\rp]\Bigg]
        ~\geq~ \E_V \Bigg[ \sum_{i\in [n]}\E_{V'}\Big[\sum_{j \in O^{V'}_i \setminus R^V_i}  \alpha e^{4\Delta}  p^{V'}_j\Big] \Bigg],
    \end{align*}
    where the inequality follows from the fact that $R^V_i \supseteq R^V_n$ for each $i \in [n]$ and that the sets $O^{V'}_1, \dots, O^{V'}_n$ are pairwise disjoint subsets of $[m]$.
    
    Now, using our MRF conditioning lemma (\Cref{lem:mrf_conditioned_bound}) and the fact that $R^V_i$ depends only on $v_1, \dots, v_{i-1}$, we have
    $\E_{V'}[ \ind{j \in O^{V'}_i \setminus R^V_i} \cdot p^{V'}_j] \geq e^{-4\Delta} \E_{V'}[ \ind{j \in O^{V'}_i \setminus R^V_i} \cdot p^{V'}_j \mid v_i' = v_i]$. Applying this to the above bound and rewriting $V'$ as $\hat V$ to reflect the conditioning, we have
    \begin{equation}\label{eq:tail-rev}
        \E_V\Big[\sum_{j\in [m]} \ind{i \text{ sold}} \cdot p_j\Big] ~\geq~ \E_V \Bigg[ \sum_{i\in [n]}\E_{\hat V}\Big[\sum_{j \in O^{\hat V}_i \setminus R^V_i}  \alpha p^{\hat V}_j ~\Big | ~ \hat v_i = v_i \Big] \Bigg].
    \end{equation}
    
    Finally, combining the revenue bound \eqref{eq:tail-rev} with the utility bound \eqref{eq:tail-util} summed over all buyers $i$, we get 
    \begin{align*}
    \E_V \Big[ \sum_{i\in [n]} v_i(S_i) \Big]
    &     = \sum_{i \in [n]} \E\Big[v_i(S_i) - \sum_{j \in S_i} p_j\Big] + \E_V\Big[\sum_{j\in [m]} \ind{i \text{ sold}} \cdot p_j\Big] \\
    &\geq \E_{V} \Bigg[ \sum_{i\in [n]}  \E_{\hat V}\Big[v_i(O_i^{\hat V}) - \alpha \sum_{j \in O^{\hat V}_i} p^{\hat V}_j + \alpha\sum_{j \in O^{\hat V}_i}(p^{\hat V}_j - e^{4\Delta}b_j)^+ ~\Big | ~ \hat v_i = v_i \Big] \Bigg]\\
    &= \sum_{i\in [n]}  \E_{V}\Bigg[v_i(O_i^{V}) - \alpha \sum_{j \in O^{V}_i} p^{V}_j + \alpha\sum_{j \in O^{V}_i}(p^{V}_j - e^{4\Delta}b_j)^+ \Bigg]\\
    &= \Big(\opt - \alpha \sum_j b_j\Big) + \alpha \cdot \E_V \Big[ \sum_{j\in [m]} \lp(p^V_j - e^{4\Delta} b_j\rp)^+\Big]. \qedhere
    \end{align*}
\end{proof}

\subsection{Core Contribution}
In this section, we prove the core contribution provided in \Cref{lem:balanced-prices-core}. To bound the core contribution, we need to exploit specific properties of the valuation class, so we will argue it separately for XOS and $k$-uniform hypergraph matching valuations.

\subsubsection{XOS Valuations}
In this setting, given the $m$ items, the $n$ buyers arrive sequentially one by one. Recall that the valuation profile $V$ is drawn from a $\Delta$-MRF distribution. Each buyer's valuation is XOS, defined as
 \begin{equation}
     v_i(S) = \max_{a \in \cA_i} \sum_{j \in S} a_j,
 \end{equation} where each $\cA_i$ is a subset of $\Rp^m$. We show, in \Cref{lem:XOS}, how to construct $(1,1)$-balanced prices that also obtain the welfare guarantee of \Cref{cor:final_guarantee_max}.

\begin{lemma}\label{lem:XOS}
   
    We define items prices $p^V_1, \dots, p^V_m$ as follows: for each $i$ and each $j \in O^V_i$, set $p^V_j = a^{(i)}_j$ where $a^{(i)} = \argmax_{a \in \cA_i} \sum_{j \in O^V_i} a_j$. Then
    \begin{enumerate}
        \item The prices $\{p^V_j\}_{j \in [m]}$ are $(1,1)$-balanced.
        \item There exists a posted-price mechanism for valuations $V \sim \cD$ such that 
        \[
        \alg \geq \frac{1}{O(\Delta)}\E_V\Bigg[ \sum_{j \in [m]} \min\{p^V_j,~e^{4\Delta}b_j\} - \frac{1}{e} \sum_{j \in [m]} b_j\Bigg].
        \]
    \end{enumerate}
\end{lemma}

\begin{proof}[Proof of \Cref{lem:XOS}]
First, we show that the prices are $(1, 1)$-balanced by checking the properties of balanced prices (\Cref{def:balanced-prices}). Property 2 clearly holds with equality as $\sum_{j \in O^V_i} p^V_j = \sum_{j \in O^V_i} a^{(i)}_j = v_i(O^V_i)$. Additionally, property 1 follows Lemma B.1 of \cite{FeldmanGL15}.

Next, we will show the following posted prices obtain the welfare guarantee of the lemma: Choose $\tau \sim \text{Unif}\{-1,0,1,\ldots,~4\Delta\}$, and set $p_j = \frac{1}{e} e^{\tau} b_j$.

To analyze the welfare of this scheme, let $\overline p^V_j : = \min\{p^V_j,~e^{4\Delta}b_j\}$ be the capped supporting price of $j$. Let $A := \{j \in [m] : \frac{p_j}{\overline p^V_j} \in [\frac{1}{e^2},~\frac{1}{e}]\}$ be the set of elements whose posted price is ``close'' to the true capped price given $V$. 
Recall that  $R_i$ is the set of items remaining when buyer $i$ arrives, $S_i$ is the set of items allocated to buyer $i$, and  $O^V_i$ is the set of items that buyer $i$ receives in the hindsight optimum  with valuations $V$.
We will bound the revenue and utility of the algorithm in terms of the items in $A$. The utility satisfies
\begin{align*}
    v_i(S_i) - \sum_{j \in S_i} p_j &\geq\quad  v_i(O^V_i \cap A \cap  R_i) - \sum_{O^V_i \cap A \cap R_i} p_j\\
    &\geq \sum_{O^V_i \cap A \cap R_i} (a^{(i)}_j - p_j) \quad \geq \quad  \sum_{O^V_i \cap A \cap R_i} (1-1/e) \cdot \overline p^V_j.
\end{align*}
Additionally, we have the revenue bound
\begin{align*}
   \E\Bigg[\sum_{i \in [n]} \sum_{j \in S_i}p_j\Bigg] ~=~ \E\Bigg[ \sum_{j \in [m]} \ind{j \text{ sold}} \cdot p_j \Bigg] ~\geq~ \frac{1}{e^2} \E \Bigg[\sum_{j \in A} \ind{j \text{ sold}}  \cdot \overline p^V_j\Bigg]  ~ \geq ~ \quad  \frac{1}{e^2} \E \Bigg[ \sum_{i\in [n]} \sum_{j \in O^V_i \cap A \setminus R_i}\overline p^V_j\Bigg].
\end{align*}
Adding together these bounds, we get
\[\E\Bigg[ \sum_{i\in [n]} v_i(S_i)\Bigg] \geq \frac{1}{e^2} \sum_{i\in [n]} \sum_{j \in O^V_i \cap A} \overline p^V_j = \frac{1}{e^2}\sum_{j \in A} \overline p^V_j.\]
Now, notice that as long as $p^V_j \geq \frac{1}{e}b_j$, we have that $\Pr_\tau[j \in A \mid V] = \frac{1}{4\Delta + 2}$. Hence, we obtain
\[\E \Bigg[ \sum_{i\in [n]} v_i(S_i)\Bigg] \geq \frac{1}{e^2(4\Delta + 2)}\Bigg(\sum_{j \in [m]} \overline p^V_j - \frac{1}{e}\sum_{j \in [m]} b_j\Bigg). \qedhere
\]
\end{proof}

\subsubsection{Hypergraph Matching}
Now we consider the online weighted matching in a $k$-uniform hypergraph with edge weights $w_i$. Similarly to the previous setting, buyers arrive online and each buyer $i$ is  interested in a specific hyper-edge $e_i\subseteq [m]$, where $|e_i| \leq k$, obtaining a random reward $w_i\geq 0$ if the edge is allocated to them. Each buyer's valuation has the form
\[
v_i(S) = w_i \cdot \ind{e_i \subseteq S}.
\]
The MRF distribution is on the random weights $w_i$ of the hyper-edges.
A hindsight optimum allocation has the form
\[
O^V_i = \begin{cases}
    e_i & \text{if }e_i \in M^V,\\
    \emptyset & \text{otherwise},
\end{cases}
\]
where $M^V$ is the max weight hypergraph matching over $E = \{e_i : i \in [n]\}$ with weights $w_{e_i} = w_i$.
In our main lemma we give $(1,k)$-balanced prices that also satisfy the requirements of \Cref{cor:final_guarantee_max}, yielding the $O(k^2(\Delta + \log k))$ competitive ratio of \Cref{thm:max_main}.

\begin{lemma}\label{lem:max_matching}
    Consider $k$-uniform hypergraph matching with $k \geq 2$ and an MRF distribution $\cD$ on the edge  weights $w_i$. 
    For a given profile $V$, define items prices $p^V_1, \dots, p^V_m$ as follows: for each $i \in [n]$ and each $j \in O^V_i$, set $p^V_j = w_i$ and for all $j \not \in \bigcup_i O^V_i$ set $p^V_j = 0$. We have
    \begin{enumerate}
        \item The prices $\{p^V_j\}_{j \in [m]}$ are $(1,k)$-balanced.
        \item There exists a posted-price mechanism for valuations $V \sim \cD$ such that 
        \[
        \alg \geq \frac{1}{O\lp(k^2(\Delta + \log k)\rp)}\Bigg(\E_V \Big[ \sum_{j\in [m]} \min\{p^V_j,~e^{4\Delta}b_j\}\Big] - \frac{1}{ek} \sum_{j\in [m]} b_j\Bigg).
        \]
    \end{enumerate}
\end{lemma}
\begin{proof}
    First, we show that our prices are $(1,k)$-balanced. We note that property 2 of \Cref{def:balanced-prices} with $\beta = k$ follows by construction, as $\sum_{j \in O^V_i} p_j^V = k \cdot w_i = k v_i(O_i)$. Additionally, property 1 follows from Lemma B.1 of \cite{FeldmanGL15}.

    Next, we give a posted-price mechanism that achieves the guarantee in the lemma. We first define ``capped'' edge-weights $\overline w_i$ for each buyer $i$ as $\overline w_i := \min \{w_i,~e^{4\Delta} \cdot \max_{j \in e_i} b_j\}$. Notice that for all $V$ and $j \in O^V_i$, we have $p^V_j \geq \overline w_i \geq \min\{p^V_j,~e^{4\Delta}b_j\}$. Thus, to show our lemma, it suffices to show that our algorithm obtains expected welfare at least
    \[
    \frac{1}{O((\Delta + \log k)k^2)}\E_V\Bigg[\sum_{i\in [n]} \ind{e_i \in M^V} \cdot \overline w_i - \frac{1}{ek}\sum_{j \in [m]} b_j\Bigg].
    \]

    To obtain our prices, we first pick $\tau \sim Unif[0, 4\Delta + \log k + 2]$. Now, for each item $j$, let $\ell_j \in \Z$ be the unique integer such that $e^{\tau +\ell_j} \in [\frac{1}{e^2k}b_j,~e^{4\Delta}b_j)$ for $b_j := \E_V \lp[ p^V_j\rp]$. Additionally, we sample $X_\ell \sim Ber(\frac{1}{k})$ independently for each $\ell \in \Z$ ; this variable indicates whether the item's price will be set to high or low. 

    Given these variables, we define our item prices as
    \[
    p_j := \begin{cases}
        \frac{1}{e} \cdot e^{4\Delta} b_j & X_{\ell_j} = 0 \text{ (``high price''}),\\
        \frac{1}{e} \cdot e^{\tau +\ell_j} & X_{\ell_j} = 1 \text{ (``low price'')}.
    \end{cases}
    \]
    First, we show that items which sell frequently at a high price achieve a high enough revenue to account for all their neighboring edges. Let $H^t := \lp\{j : \Pr[j \text{ sells} \mid X_{\ell_j} = 0,~ \tau = t] \geq \frac{1}{k^2}e^{-4\Delta - 2}\rp\}$ and denote by $\delta(H^t)$ the neighboring edges of the items in the set $H^t$. 
    \begin{restatable}{claim}{HighSellingClaim}\label{cl:high_selling}
        For fixed $t$, we have $\E \big[\sum_{j \in H^t} \ind{j \text{ sells}} \cdot p_j  \mid \tau = t\big] \geq \frac{1}{2e^3k^2} \sum_{\substack{e_i \in \delta(H^t)}} \ind{e_i \in M^V} \cdot w_i$    
    \end{restatable}

    \begin{proof}[Proof of \Cref{cl:high_selling}]
        For $j \in H^t$, we have 
        \begin{align*}
            \E [\ind{j \text{ sells}} \cdot p_j  \mid \tau = t] 
            &\geq \E \lp[\ind{j \text{ sells}}\cdot e^{4\Delta -1} b_j  \mid X_{\ell_j} = 0,~\tau = t\rp] \Pr[X_{\ell_j} = 0 \mid \tau = t]\\
            &\geq \frac{1}{e^{4\Delta +2}k^2} \cdot e^{4\Delta - 1}b_j \cdot \frac{k-1}{k} \quad \geq \quad  \frac{b_j}{2e^3k^2}.
        \end{align*}
        Using $b_j = \E \lp[ \sum_{e_i \in \delta(j)} \ind{e_i \in M^V} w_i\rp]$ and summing over $j \in H^t$, we get our claim.
            \end{proof}

    Hence, it suffices to focus on the contribution from buyers $i$ such that $e_i \not \in \delta (H^\tau)$, as taking expectation over $\tau$ an applying \Cref{cl:high_selling} ``covers'' the value for $\delta(H^\tau)$. For such a buyer $i$ with $e_i \not \in \delta(H^\tau)$, define the $\xi_i$ to be the event that
    \begin{enumerate}
        \item $X_{\ell^i_{\max}} = 1$, where $\ell^i_{\max} = \max_{j \in e_i} \ell_j$.
        \item $X_{\ell_j} = 0$ for all $j \in e_i$ with $\ell_j \neq \ell^i_{\max}$.
    \end{enumerate}

    In other words, $\xi_i$ is the event that the items $j \in e_i$ with large $b_j$ are all ``low-priced,'' while all other items are ``high-priced.'' Let $B^i_{low} = \{j \in e_i : X_{\ell_j} = 1\}$ and $B^i_{high} = \{j \in e_i : X_{\ell_j} = 0\}$ denote these low and high priced items respectively. Notice that even in the event $\xi_i$, all items in $e_i$ have price at most $\frac{1}{e}e^{\tau + \ell^i_{\max}}$. Additionally, we note that $\xi_i$ is independent of the valuation $V$, and 
    \[
    \Pr[\xi_i] \geq \frac{1}{k} \cdot \left(\frac{k-1}{k}\right)^{k-1} \geq \frac{1}{ek}.
    \]

    \noindent
    Now, we bound the utility and the revenue separately for each fixed value of $\tau$. 
    \paragraph{Bounding the utility.} Recall that we denote by $R_i^V$ the set of items that remain when buyer $i$ arrives, for a given $V$. We have
    \begin{align*}
        \E\Big[v_i(S_i) - \sum_{j \in S_i} p_j\Big] &\geq \E_V\Bigg[\ind{\xi_i}\ind{e_i \subseteq R^V_i} \cdot \Big(w_i - \sum_{j \in e_i} p_j\Big)^+ \Bigg] & \\
        &\geq \E_V\Bigg[\ind{\xi_i}\ind{e_i \subseteq R^V_i} \cdot \Big(w_i - \frac{1}{e}e^{\tau +\ell^i_{\max}}\Big)^+\Bigg] & \\
        &\geq \E_V\left[\ind{\chi_i}\ind{\xi_i}\ind{e_i \subseteq R^V_i} \cdot \frac{e-1}{e} \cdot \overline w_i\right] &  &\\
        &\geq \E_V\left[\ind{\chi_i}\ind{\xi_i} \lp(\ind{B^i_{high} \subseteq R^V_i} - \ind{B^i_{low} \not \subseteq R^V_i}\rp) \cdot \frac{e-1}{e} \cdot \overline w_i\right] .
    \end{align*}
    Here, $\chi_i$ denotes the event that $\frac{e^{\tau +\ell^i_{\max}}}{\overline w_i} \in [\frac{1}{e},1]$. We show the following claim on the term that corresponds to the high-priced items. 
    \begin{restatable}{claim}{HighPricedClaim}\label{cl:high_priced}
        If $i \not \in \delta(H^\tau)$, then we have
        \[
        \E\left[\ind{\chi_i} \ind{\xi_i} \ind{B^i_{high} \subseteq R^V_i} \cdot \overline w_i\right] \geq (1 - 1/e)\E\left[\ind{\chi_i} \ind{\xi_i}\cdot \overline w_i\right] .
        \]
    \end{restatable}
  \begin{proof}[Proof of Claim~\ref{cl:high_priced}]
        It suffices to show that $\Pr[B^i_{high} \not\subseteq R^V_i \mid \xi_i,~\chi_i,~\overline w_i] \leq 1/e$. We have
        \begin{align*}
            \Pr[B^i_{high} \not\subseteq R^V_i \mid \xi_i,~\chi_i,~\overline w_i] 
            &\leq \sum_{j \in B^i_{high}} \Pr[j \not \in R_i \mid \xi_i,~\chi_i,~\overline w_i]\\
            &= \sum_{j \in B^i_{high}} \Pr[j \not \in R_i \mid \xi_i,~\overline w_i]\\
            &= \sum_{j \in B^i_{high}} \frac{\Pr[(j \not \in R_i)  \wedge \xi_i \mid \overline w_i]}{\Pr[\xi_i \mid \overline w_i]}\\
            &\leq ek\sum_{j \in B^i_{high}} \Pr[j \not \in R_i \mid \overline w_i]\\
            &\leq ek\sum_{j \in B^i_{high}} e^{4\Delta} \Pr[j \not \in R_i] \quad \leq \quad  ek \cdot k e^{4\Delta} \cdot \frac{1}{k^2} e^{-4\Delta - 2} \quad = \quad  1/e.       \qedhere
        \end{align*}
    \end{proof}

    Applying our claim to the utility bound and summing over just $i$ with $e_i \in M^V$, we get
    \begin{align*}
        \sum_{i\in [n]} \E\Big[v_i(S_i) - \sum_{j \in S_i} p_j\Big] &\geq \left(\frac{e-1}{e}\right)^2\E_V \sum_{e_i \in M^V \setminus \delta(H)}\left[\ind{\chi_i}\ind{\xi_i} \cdot \overline w_i\right] - \E_V \sum_{e_i \in M^V}\left[\ind{\chi_i}\ind{\xi_i}\ind{B^i_{low} \not \subseteq R^V_i} \cdot \overline w_i\right].
    \end{align*}

\paragraph{Bounding the revenue.}
We have
\begin{align*}
    \E \Big[ \sum_{j\in [m]} \ind{j \text{ sold}} \cdot p_j \Big]
    &\geq \E \Big[ \sum_j \ind{X_{\ell_j} = 1} \ind{j \text{ sold}} \cdot \frac{1}{e}e^{\tau + \ell_j}\Big]\\
    &= \E \Big[ \sum_{e_i \in M^V} \sum_{j \in e_i} \ind{X_{\ell_j} = 1} \ind{j \text{ sold}} \cdot \frac{1}{e}e^{\tau + \ell_j}\Big]\\
    &= \E \Big[ \sum_{e_i \in M^V} \sum_{j \in B^i_{low}} \ind{j \text{ sold}} \cdot \frac{1}{e}e^{\tau + \ell_j} \Big]\\
    &\geq \E \Big[ \sum_{e_i \in M^V} \sum_{j \in B^i_{low}} \ind{\xi_i} \ind{j \text{ sold}} \cdot \frac{1}{e}e^{\tau + \ell^i_{\max}} \Big]\\
    &\geq \E \Big[ \sum_{e_i \in M^V} \ind{\xi_i} \ind{B^i_{low} \not \subseteq R^V_i} \cdot \frac{1}{e}e^{\tau + \ell^i_{\max}} \Big]\\
    &\geq \E \Big[ \sum_{e_i \in M^V} \ind{\chi_i} \ind{\xi_i} \ind{B^i_{low} \not \subseteq R^V_i} \cdot \frac{1}{e^2}\overline w_i \Big].
\end{align*}

    Hence, combining our revenue and utility bounds at the appropriate ratios, we have
    \begin{align*}
        \E \Big[\sum_{i \not \in \delta(H^\tau)} v_i(S_i) \Big]&\geq \frac{1}{e^2}\left(\frac{e - 1}{e}\right)^2\E_V \Big[\sum_{e_i \in M^V \setminus \delta(H^\tau)}\ind{\chi_i}\ind{\xi_i} \cdot \overline w_i\Big]\\
        &\geq \frac{1}{e^3k}\left(\frac{e - 1}{e}\right)^2\E_V \Big[\sum_{e_i \in M^V \setminus \delta(H^\tau)}\ind{\chi_i}\cdot \overline w_i\Big].
    \end{align*}
    Then, by adding contribution from buyers $i\in \delta(H^\tau)$ from \Cref{cl:high_selling}, we have
 \[
 \E \Big[ \sum_{i\in [n]} v_i(S_i) \Big]\geq \frac{1}{e^3k^2}\left(\frac{e - 1}{e}\right)^2 \cdot \E_V \Big[ \sum_{e_i \in M^V}\ind{\chi_i}\cdot \overline w_i\Big].
 \]
 
    Finally, we now take the expectation over the choice of $\tau$. Notice that given $V$ (and hence, $\overline w_i$), the conditional probability of $\chi_i$ is $\frac{1}{4\Delta + \log k +2}$ as long as $w_i \geq \frac{1}{ek}\max_{j \in e_i}b_j$. Using this, we get
    \begin{align*}
        \E \Big[ \sum_{i} v_i(S_i) \Big]&\geq \frac{1}{e^3k^2}\Big(\frac{e - 1}{e}\Big)^2 \cdot \E_V \Bigg[ \sum_{e_i \in M^V}\frac{1}{4\Delta + \log k + 2}  \left(\overline w_i - \frac{1}{ek}\max_{j \in e_i} b_j\right)\Bigg]\\
        &\geq \frac{1}{O((\Delta + \log k) k^2)}\E_V\Big[\sum_{e_i \in M^V} \overline w_i - \frac{1}{ek}\sum b_j\Big] . \qedhere
    \end{align*}    
\end{proof}

%% file: lowerBounds.tex
\section{Lower Bounds for Minimization Problems} \label{sec:lowerBounds}

Our algorithms for  minimization problems in \Cref{sec:minimization} give competitive ratios that are linear in $\Delta$. A natural question is whether better competitive ratios are possible. 
We answer this question by giving a general lower bound technique for any subadditive coverage problem. To achieve this, we first show that  $\Delta$-MRFs can very accurately capture \emph{time-dependent Markov Chains} (\Cref{def:mc}). Then,  \Cref{lem:general-min-hardness}  shows how to transform any hardness that holds for a time-dependent Markov Chain into a hardness result for a $\Delta$-MRF instance, by only losing a constant factor. Using known worst-case lower bounds for Facility Location and Steiner Tree,  
we   apply this lemma to 
 obtain  an $\Omega\lp({\Delta}/{\log \Delta}\rp)$ lower bound for Facility Location (Corollary~\ref{cor:LB_FL}) and 
 an $\Omega(\Delta)$ lower bound for  Steiner Tree (Corollary~\ref{cor:LB_st_tree}).

\begin{definition}[Time-Dependent Markov Chain]\label{def:mc}
    A random variable $X = (X_1, \dots, X_n) \in \Omega$, where $\Omega = \Omega_1 \times \dots \times \Omega_n$, is a \emph{time-dependent Markov chain} if for any $\omega = (\omega_1, \dots, \omega_n) \in \Omega$ and any $i \geq 2$, 
    \[
    \Pr[X_i = \omega_i \mid \forall j < i,~ X_j = \omega_j] = \Pr[X_i = \omega_i \mid X_{i-1} = \omega_{i-1}].
    \]
\end{definition}

In the following lemma we show that for any arbitrary time-dependent Markov Chain there exists a $\Delta$-MRF that is arbitrarily close to it.

\begin{lemma}\label{lem:mrf-mc-coupling}
    Let $X = (X_1, \dots, X_n) \sim \cD$ be an arbitrary time-dependent Markov chain over $\Omega = \Omega_1 \times \dots \times \Omega_n$. Then, for any $\epsilon > 0$, there is a $\Delta$-MRF distribution $\cD_\cM$ such that for any $\omega \in \Omega$, 
    \[
    \Pr_{Y \sim \cD_\cM} [Y = \omega] \geq (1 - \epsilon) \cdot \Pr_{X \sim \cD} [X = \omega],
    \]
    where $\Delta = O(\log \frac{n \cdot \max_i  |\Omega_i|}{\epsilon})$.
\end{lemma}
\begin{proof}[Proof of \Cref{lem:mrf-mc-coupling}]
    Let $\Delta \geq 0$ be a value chosen later. For each $i \in [n-1]$, $s \in \Omega_i$, and $t \in \Omega_{i+1}$, let 
    \[
    \delta_{i,i+1}(s,t) := \log \left(\Pr[X_{i+1} = t \mid X_i=s]\right) \leq 0.
    \]

    Consider an MRF $\cM = (\Omega, E, \{\psi_i\}_i,\{\psi_e\}_e\}$ and corresponding random variable $Y := (Y_1, \dots, Y_n) \sim \cD_\cM$ given by (also see \Cref{fig:MarkovChain})
    \begin{itemize}
        \item $E := \{\{i,i+1\} : i \in [n-1]\}$.
        \item $\psi_{i,i+1}(s, t) =
            \max\{\delta_{i,i+1}(s,t),~-\Delta/2\}.$
    \end{itemize}

   \newcommand{\equal}{=} 
\tikzstyle{point}=[circle, draw, fill=black!30, inner sep=0pt, minimum width=2pt]
\tikzstyle{graphnode}=[circle, draw, fill=black!30, inner sep=0pt, minimum width=12pt]
\tikzstyle{localgraphnode}=[circle, draw, fill=red!100, inner sep=0pt, minimum width=12pt]
\begin{figure}[ht]
\begin{center}
\begin{tikzpicture}[thin,scale=1]

	\foreach \y in {-4,-2,2,4}{
		\node at (\y,-2) [graphnode]{};
	}
    \draw [thick] (-2-1.8,-2) to (-2.2,-2);
    \draw [thick] (-1.8,-2) to (-1,-2);
    \draw [thick] (1,-2) to (1.8,-2);
    \draw [thick] (2.2,-2) to (3.8,-2);
	\node at (-4,-2) [label=above:$Y_1\equal X_1$]{};
	\node at (-2,-2) [label=above:$Y_2$]{};	
	\node at (2,-2) [label=above:$Y_{n-1}$]{};
	\node at (4,-2) [label=above:$Y_n$]{};
	\foreach \y in {-0.3,0,0.3}{
		\node at (\y,-2) [point]{};	
	}			
\end{tikzpicture}
\end{center}
\caption{MRF $\cM = (\Omega, E, \{\psi_i\}_i,\{\psi_e\}_e\}$ with random variables $Y_1, \dots, Y_n$.}
\label{fig:MarkovChain}
\end{figure}
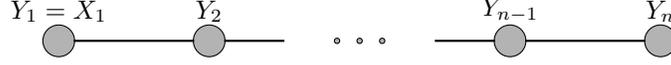

    It is not hard to see that the weighted max degree of $\cM$ is at most $\Delta$, as $|\psi_{i-1,i}(s,t)| \leq \Delta/2$. We  choose $\psi_1$ so that $Y_1$ is identically distributed to $X_1$.
    For each $i \geq 2$, we choose $\psi_i$ so that the marginal distribution of $Y_i$ given $Y_{i-1} = s$ is simply $\Pr[Y_i = t \mid Y_{i-1} = s] =\frac{\exp(\psi_{i-1,i}(s,t))}{\sum_{u \in \Omega_i}\exp(\psi_{i-1,i}(s,u))}$ (see Lemma 5.3 of \cite{VPS24} for an explicit construction of such $\psi_i$).

    \begin{claim}
        For $i \geq 2$, $s \in \Omega_{i-1}$, and $t \in \Omega_{i}$, we have
        \[
        \Pr[Y_{i} = t \mid Y_{i-1} = s] \geq \Pr[X_{i} = t \mid X_{i-1} = s] \cdot (1 - |\Omega_i|\cdot e^{-\Delta/2}) .
        \]
    \end{claim}
    \begin{proof}
        By construction,
        \begin{align*}
        \Pr[Y_{i} = t \mid Y_{i-1} = s] &= \frac{\exp(\psi_{i-1,i}(s,t))}{\sum_{u \in \Omega_i}\exp(\psi_{i-1,i}(s,u))} \\
        &\geq \frac{\exp(\delta_{i-1,i}(s,t))}{|\Omega_i|e^{-\Delta/2} + \sum_{u \in \Omega_i}\exp(\delta_{i-1,i}(s,u))}
        = \frac{\Pr[X_i = t \mid X_{i-1} = s]}{|\Omega_i|e^{-\Delta/2} + 1}, 
        \end{align*}
        where the inequality uses $\exp({\psi_{i-1,i}(s, u)}) \leq
            \exp({\delta_{i-1,i}(s,u)})  + \exp({-\Delta/2})$.
    \end{proof}
    Intuitively, this claim implies that $X$ and $Y$ have very similar transition probabilities between states as long as $|\Omega_i| \cdot e^{-\Delta/2}$ is small for each $i$. 
    Using this claim, we get for any $\omega = (\omega_1, \dots, \omega_n) \in \Omega$,
    \begin{align*}
        \Pr[Y = \omega] &= \Pr[Y_1 = \omega_1] \cdot \prod_{i =2}^n \Pr[Y_i = \omega_i \mid Y_{i-1}]\\
        &\geq \Pr[X_1 = \omega_1] \cdot \prod_{i = 2}^n \Pr[X_i = \omega_i \mid X_{i-1} = \omega_{i-1}](1 - |\Omega_i|\cdot e^{-\Delta/2}) \\
        &\geq \Pr[X = \omega] \cdot (1 - n \cdot \max_i |\Omega_i| \cdot e^{-\Delta/2}).
    \end{align*}
    By setting $\Delta = 2 \log \left(\frac{n \max_i \cdot |\Omega_i|}{\epsilon}\right)$, we get $\Pr[Y = \omega] \geq (1 - \epsilon) \Pr[X = \omega]$, as desired in \Cref{lem:mrf-mc-coupling}.
\end{proof}

We now use this lemma to obtain a general lower bound for online subadditive coverage problems in the $\Delta$-MRF prophet model.

\begin{lemma}\label{lem:general-min-hardness}
    Let $M, C \geq 0$, $\alpha \geq 1$. Suppose for some subadditive coverage problem, there is a time-dependent Markov chain $X =(X_1, \dots, X_n) \sim \cD$ over $\Omega = \Omega_1 \times \dots \times \Omega_n$ such that when arrivals are sampled from $\cD$,  for any online algorithm $\E_{X \sim \cD} [\alg(X)] \geq \alpha \cdot \E_{X \sim \cD} [\opt(X)]$. Furthermore, suppose $|\Omega_i| \leq M$ for each $i$ and $\opt(\omega) \leq C$ for any $\omega \in \Omega$. Then, there exists a $\Delta$-MRF distribution $\cD_\cM$ such that any online algorithm has
    \[
    \E_{Y \sim \cD_\cM} [\alg(Y)] \geq \alpha/2 \cdot \E_{X \sim \cD_\cM} [\opt(X)],
    \]
    where $\Delta = O\left(\log \frac{nMC}{\E [\opt(X)]}\right)$.
\end{lemma}
\begin{proof}
    Let $\epsilon > 0$ be a value chosen later. Using \Cref{lem:mrf-mc-coupling}, we may construct an MRF $\cM$ with max degree at most $\Delta = 2 \log \frac{n M}{ \epsilon}$ such that for any $\omega \in \Omega$, we have
    \[
    \Pr_{Y \sim \cD_\cM} [Y = \omega] \geq (1 - \epsilon)\Pr_{X \sim \cD}[X = \omega].
    \]
    Using this, we will compute bounds on $\E_{\cD_\cM} \alg(Y)$ and $\E_{\cD_\cM} \opt(Y)$. To do so, we define for each $\omega \in \Omega$ the value $q_\omega := \Pr[Y = \omega] - (1-\epsilon) \Pr[X = \omega]$. Notice that $q_\omega \geq 0$ and $\sum_{\omega \in \Omega} q_\omega = \epsilon$.  Now, we have 
    \begin{align*}
        \E_{\cD_\cM} [\opt(Y) ]
        ~=~ \sum_{\omega \in \Omega} \Pr[Y = \omega] \cdot \opt(\omega)
        &=~ (1 - \epsilon)\sum_{\omega \in \Omega} \Pr[X = \omega] \cdot \opt(\omega) + \sum_{\omega \in \Omega} q_\omega \opt(\omega)\\
        &=~ (1 - \epsilon) \E_{\cD}[\opt(X)] + \sum_{\omega \in \Omega} q_\omega \opt(\omega).
    \end{align*}
    Additionally, 
    \begin{align*}
        \E_{\cD_\cM} [\alg(Y)] 
        ~=~ \sum_{\omega \in \Omega} \Pr[Y = \omega] \cdot \alg(\omega) ~&=~ (1 - \epsilon)\sum_{\omega \in \Omega} \Pr[X = \omega] \cdot \alg(\omega) + \sum_{\omega \in \Omega} q_\omega \alg(\omega)\\
        &\geq~ (1 - \epsilon) \alpha \E_{\cD}[\opt(X)] + \sum_{\omega \in \Omega} q_\omega \opt(\omega).
    \end{align*}
    Hence, noticing that $\sum_{\omega \in \Omega} q_\omega \opt(\omega) \
    \leq \epsilon C$, we have
    \begin{align*}
        \frac{\E_{\cD_\cM} [\alg(Y)]}{\E_{\cD_\cM} [\opt(Y)]} 
        &\geq \frac{(1 - \epsilon) \alpha \E_{\cD}[\opt(X)] + \sum_{\omega \in \Omega} q_\omega \opt(\omega)}{(1 - \epsilon) \E_{\cD}[\opt(X)] + \sum_{\omega \in \Omega} q_\omega \opt(\omega)}\\
        &\geq \frac{(1 - \epsilon) \alpha \E_{\cD}[\opt(X)] + \epsilon C}{(1 - \epsilon) \E_{\cD}[\opt(X)] + \epsilon C} \quad =\quad  \frac{\alpha + \eta}{1 + \eta},
    \end{align*}
    where $\eta = \frac{\epsilon C}{(1-\epsilon)\E_{\cD} [\opt(X)]}$. Choosing $\epsilon = \frac{\E_{\cD} [\opt(X)]}{2C} \leq \frac{1}{2}$ gives $\eta \leq 1$, and hence $\frac{\alpha + \eta}{1 + \eta} \geq \frac{\alpha}{2}$, as desired.
\end{proof}

\begin{corollary}\label{cor:LB_FL}
    For the problem of online facility location with $\Delta$-MRF arrivals, there is an instance for which no online algorithm is better than $\Omega(\Delta/\log \Delta)$-competitive.
\end{corollary}
\begin{proof}
    We note that the $\Omega\lp(\frac{\log n}{\log \log n}\rp)$ hard instance for online facility location given by \cite{Fota2008} satisfies the conditions of \Cref{lem:general-min-hardness}. First, the distribution over arrival sequences $X_1, \dots, X_n$ is a time-dependent Markov chain, as the distribution for arrival $X_i$ is determined completely by the $X_{i-1}$. We also note that $\max_i |\Omega_i| = O(n)$ and $\opt(\omega) \leq n \cdot \E [\opt(X)]$ (since for any sequence of demands $\omega$, one can always open a facility at every demand location). Therefore, \Cref{lem:general-min-hardness} tells us that there is an $\Omega\big(\frac{\log n}{\log \log n}\big)$ hard instance in the $\Delta$-MRF with $\Delta = O(\log (n^3)) = O(\log n)$. Sending $n \to \infty$, this gives a family of $\Omega\big(\frac{\Delta}{\log \Delta}\big)$ hard instances for $\Delta \to \infty$.
\end{proof}

\begin{corollary}\label{cor:LB_st_tree}
    For the problem of online Steiner tree with $\Delta$-MRF arrivals, there is an instance for which no online algorithm is better than $\Omega(\Delta)$-competitive.
\end{corollary}
\begin{proof}
We note that a version of the  ``diamond graph'' $\Omega(\log n)$ hard instance for online Steiner tree, described in \cite{MakoWaxm1991, GargGuptLeonSank2008} satisfies the conditions of \Cref{lem:general-min-hardness}, although we must be careful about the arrival order to ensure the distribution satisfies the Markov chain property. We repeat this instance for completeness.

In the hardness instance, our underlying graph is generated inductively for $k \geq 0$ as follows: Define $G_0 = K_2$. For each $i \in [k]$ define $G_i$ as the graph which results from replacing each edge $uv \in E(G_{i-1})$ with two disjoint paths of length two, i.e., remove edge $uv$ and add new vertices $x,y$ along with new edges $ux, xv,uy,yv$. Each pair of vertices $x,y$ born from the same edge $uv \in E(G_{i-1})$ in this way are called \emph{twins}. For each $v \in V(G_k)$, the \emph{rank} of $v$ is the first $i$ in which $v \in V(G_i)$. Moreover, the \emph{parent} of each vertex $v \in V(G_k)$ of rank $i \geq 2$ is defined as the unique node $p$ of rank $i-1$ adjacent to $v$ in $G_i$. If $p$ is the parent of $v$, we call $v$ the \emph{child} of $p$. Notice that each vertex $v \in V(G_k)$ of rank $i \in [k-1]$ has two pairs of twin children.

\IGNORE{
 \begin{figure}
     \centering
     \includegraphics[width=0.5\linewidth]{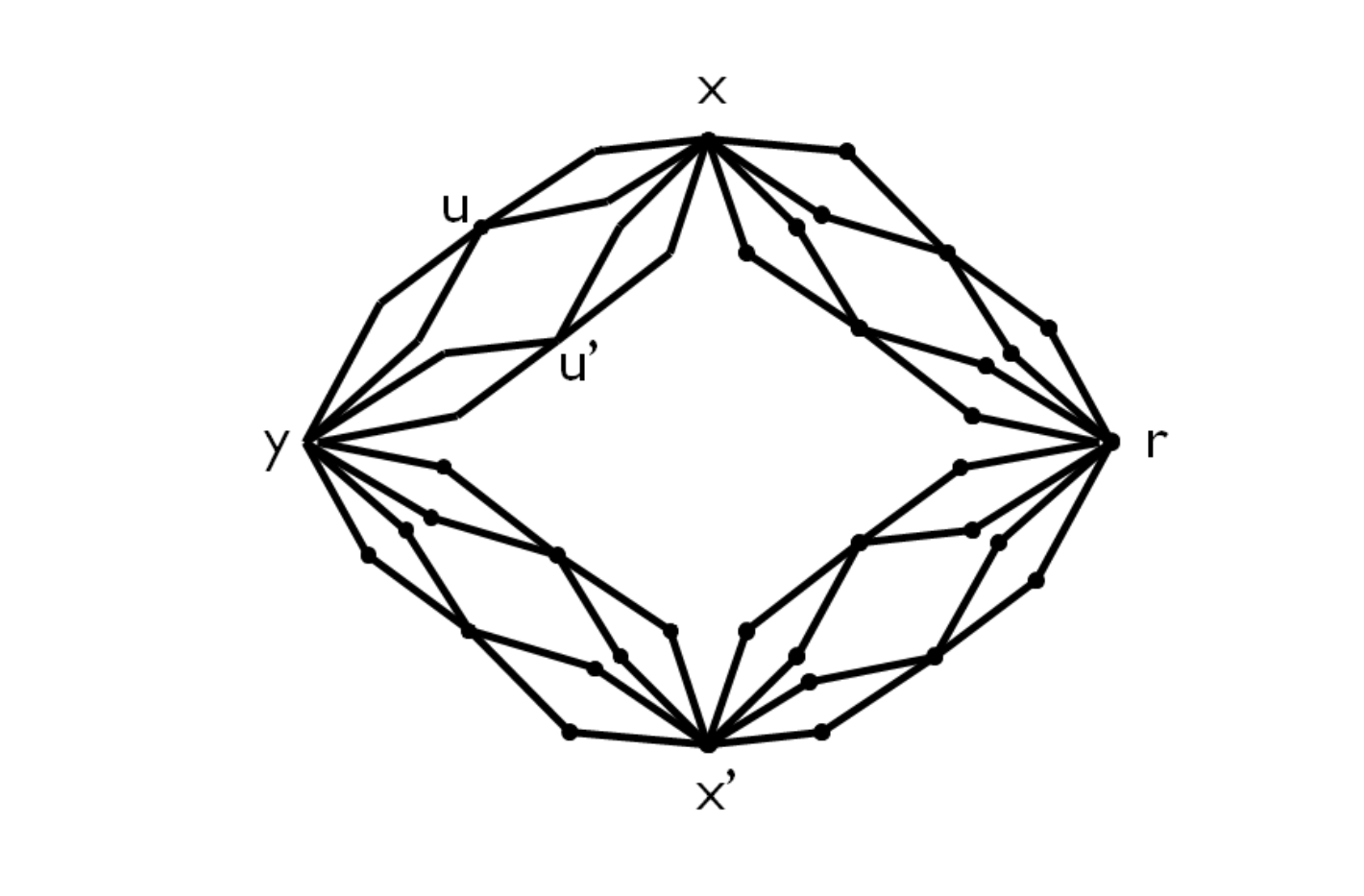}
     \caption{Caption}
     \label{fig:enter-label}
 \end{figure}
}

For our instance of online Steiner tree, our underlying graph is $G_k$. We set one vertex of rank $0$ as the root $r$, and the other vertex of rank $0$ as the first arrival $X_1$. The second arrival $X_2$ is chosen uniformly from one of the two twin vertices of rank $1$. For the remainder of arrivals, the adversary does the following: Let $\{x,y\}$ be the pair of twin vertices closest to $r$ in $G_k$ for which 
\vspace{-0.15cm}
\begin{enumerate}
    \item neither $x$ nor $y$ has arrived yet as a demand point, and
    \item the parent of $x$ and $y$ has already arrived as a demand point.
\end{enumerate}
\vspace{-0.15cm}
Then, pick one of either $x$ or $y$ uniformly at random to arrive next. If no such pair $(x,y)$ exists, which occurs after $2^k + 2 =: n$ arrivals, then halt.

From \cite{MakoWaxm1991}, we know that no online algorithm can obtain competitive ratio better than $\Omega(k) = \Omega(\log n)$ on this instance. Moreover, we crucially note such a distribution can be constructed as a time-dependent Markov chain due to the ``depth-first traversal'' of arrivals (i.e. always picking the arrival $X_i$ from the twin pair $(x,y)$ closest to $r$, which ensures $(x,y)$ can be uniquely determined from the prior arrival $X_{i-1}$). Additionally, we have $\max_i |\Omega_i| \leq |V(G_k)| = O(4^k) = O(n^2)$, and $\frac{\max_{\omega \in \Omega} \opt(\omega)}{\E \opt(X)} \leq n$. Therefore, \Cref{lem:general-min-hardness} implies that there exists a $\Delta$-MRF hardness of $\Omega(\log n)$ for $\Delta = O(\log(n \cdot n^2 \cdot n)) = O(\log n)$. Hence, taking $n \to \infty$ gives a family of $\Omega(\Delta)$-hard instances.
\end{proof}

%% file: appendix.tex
\section{Appendix}

\subsection{Hardness for $p$-Sample Maximization Problems}\label{apn:hardness}

In this section, we show that our approach for minimization problems via reduction to the $p$-sample independent model is not interesting for  maximization problems. 
In the minimization setting, our reduction relies on the existence of an $O(\log(1/p))$-competitive algorithm for $p$-sample independent instances. 
However, below we will construct a prophet inequality (single-item allocation problem) instance where any algorithm restricted to a $p$-fraction of independently sampled items suffers an $\Omega(1/p)$ approximation loss. This implies that there does not exist an $O(\log(1/p))$-competitive algorithm for this specific problem. Given that there is known $O(\Delta)$-competitive algorithm~\cite{VPS24} for prophet inequality under MRFs, we conclude that this reduction framework is unlikely to yield meaningful algorithms for maximization problems under MRFs.

\begin{restatable}{lemma}{WhyMaxBreaks}\label{lem:why_max_breaks}
The \( p \)-sample prophet inequalities, where a \( p \)-fraction of the instance is independently and uniformly sampled and revealed to the online algorithm upfront, is \( \Omega \left( \frac{1}{p} \right) \)-hard to approximate.
\end{restatable}

\begin{proof}[Proof of Lemma~\ref{lem:why_max_breaks}]

    Consider the following instance: for a fixed  $M \gg n$, let \( X_1, X_2, \dots, X_n \) be a sequence of random variables generated as $ X_1 = 1$ and, for \( i \geq 2 \), the value \( X_i \) is
    \[
    X_i =
    \begin{cases}
        M \cdot X_{i-1}, & \text{with probability } \frac{1}{M}, \\
        0, & \text{otherwise}.
    \end{cases}
    \]

\noindent    This defines a sequence  $(1, M, M^2, \dots, M^{k-1}, 0, 0, \dots, 0)$ with optimal stopping point $X_{k}$ and $\opt=M^{k-1}$.

First, consider this instance with no sample. Note that we have
    $\Pr[X_{t + 1} = M^{t} \mid X_t = M^{t - 1}] = \frac{1}{M}.$
    Let \( \alg_i \) denote the value obtained by the algorithm starting at \( X_i \), given that the value of \( X_i \) is \( M^{i - 1} \), then $\E[\alg_{n}] = M^{n - 1}$.
    
    Suppose at time \( t \), the algorithm is at \( X_t \) with value \( M^{t - 1} \), and it is given that
    $
    \E[\alg_{t + 1}] = M^t
    $.
    If the algorithm stops at \( X_t \), the value obtained is \( M^{t - 1} \). If it proceeds to \( X_{t + 1} \), then:
    \[
    \E[\alg_{X_t \text{ proceeds}}] = \frac{1}{M} \cdot \E[\alg_{t + 1}] + \left(1 - \frac{1}{M} \right) \cdot 0
    = \frac{1}{M} \cdot M^{t} + \left(1 - \frac{1}{M} \right) \cdot 0
    = M^{t - 1}.
    \]
    Hence, regardless of the algorithm's decision, we have
    $\E[\alg_{t}] = M^{t - 1}$.   
    By induction, we have 
    \[
    \E[\alg_t] = M^{t - 1}, \quad \forall 1 \leq t \leq n.
    \]
    In particular, for \( t = 1 \), we have
    $
    \E[\alg_1] = 1,
    $
    which is equivalent to simply taking the first value.

    Now, we return to consider the setting where a $p$-fraction of values $X_i$ are revealed up front as a sample. Suppose $k^\prime$ is the index of the the largest non-zero sample $X_{k^{\prime}} = M^{k^\prime - 1}$. It is not hard to see that the optimal online algorithm should wait at least until item $X_{k^\prime}$ before stopping. Moreover, since each sample revealed after $X_{k^\prime}$ is $0$, the conditional distribution of values $X_{k^\prime + 1}, \dots, X_n$ given the sample is only worse than it is in the setting of $\alg_{k^\prime}$ above.
    Hence, the algorithm cannot obtain expected value better than $\alg_{k^\prime}$, so $\E[\alg \mid k'] \leq \alg_{k^\prime} = M^{k^{\prime} - 1}$. Therefore, we have
    \begin{align*}
    \E[\alg] 
     \quad = \quad  \sum_{i = 1}^{n}{\Pr[k^{\prime} = i]} \cdot {M^{i - 1}}
     \quad \leq \quad \sum_{i = 1}^{n}
   \frac{p}{M^{i-1}} \cdot M^{i-1}
     \quad = pn.
   \nonumber
     \end{align*}

    Since, $\opt = \sum_{i=1}^n \left(\frac{M-1}{M}\right)^{n-i} = n \cdot (1 - o(1))$ for $M \gg n$, this gives a lower bound on the competitive ratio of $\Omega \left( \frac{1}{p} \right)$.
\end{proof}

Since there is an \( \Omega (\frac{1}{p})\) lower bound on the competitive ratio for prophet inequalities with \( p \)-sample independent instances, the reduction used for cost minimization problems, which requires an \( O\left(\log \left( \frac{1}{p} \right) \right) \)-competitive algorithm for \( p \)-sample independent instances, is not useful for the maximization problem.

%% file: facilityLocation.tex
\subsection{Monotone Algorithms in the $p$-Sample Independent Model} \label{sec:monotoneAlgos}

We show why existing $p$-sample independent algorithms (or their variants) for Steiner Tree and Facility Location are monotone. Recall, to prove that an algorithm is monotone, we need to consider a sample set $S$ that is formed by first independently including each value in $V$ independently with probability $p$ and then further augmenting it by moving some real values into the sample set. In other words, the distribution on $S$ stochastically dominates the product distribution containing each value independently with probability $p$.

\subsubsection{Steiner Tree}
We have $n$ demands on a metric space $(V,E,d)$ with a fixed root $r\in V$. Each demand $x_i\in V$ has to be connected to the root by purchasing a set of edges, where edge $e \in E$ has a known cost $c(e)$. The goal is to minimize the total cost of purchased edges.

We observe that the $O(\log (1/p)$-competitive greedy algorithm of \cite{GargGuptLeonSank2008,ArgyFrieGuptSeil2022} is monotone.

\begin{theorem} \label{thm:SteinerTreeMonot}
    Algorithm~\ref{algo:StTree-gen} is a monotone $O(\log(1/p))$-competitive to the optimal offline $\opt(V)$. 
\end{theorem}

\begin{algorithm}[H]
\KwIn{$p$-sample independent}
\tcp{Phase 1}
Receive the sample set $S$ and obtain an MST on $S$.\\
\tcp{Phase 2}
For the set of real arrivals $R$, greedily connect the next arrival $x_i \in R$ to the closest point in $S$.
    \caption{Monotone Algorithm for Online Steiner Tree}\label{algo:StTree-gen}
\end{algorithm}
 
\begin{proof}
The algorithm is monotone since augmenting the sample set can only make the connection costs smaller for the real arrivals. Hence,  the proof follows by \cite{ArgyFrieGuptSeil2022}.
\end{proof}

\subsubsection{Facility Location}
We have $n$ demands on a metric space that will arrive one by one. Each demand $x_j\in C$ has to either connect to an open facility or pay the connection cost. Denote by $F$ the set of open facilities, $f$ the (uniform) connection cost and $d(x_j, F)=\min_{c\in F} d(c,x_j)$ the distance of demand $x_j$ to the closest open facility. We want to minimize
\[ \E\Big[ |F|\cdot f + \sum_{j\in [n]} d(F, x_j) \Big].\]

We give a monotone $O(\log 1/p)$-competitive algorithm for this problem.
\begin{theorem} \label{thm:facilLocMonot}
    Algorithm~\ref{algo:FL-gen} is a monotone $O(\log(1/p))$-competitive to the optimal offline $\opt(V)$. 
\end{theorem}

\begin{algorithm}[H]
\KwIn{$p$-sample independent, cost $f$}
\tcp{Phase 1}
Receive the sample set $S$ and run any $O(1)$-competitive offline algorithm on $S$, denote by $\hat F$ the facilities opened.\\
\tcp{Phase 2}
For the set of real arrivals $R$, use Meyerson's algorithm \cite{Meyerson2001}: \newline 
Open a new facility with probability $\min\{ d(x_i,F_{i-1})/f , 1 \}$, where $F_i$ is the set of open facilities (including  $\hat F$) after the $i$-th arrival.
    \caption{Monotone Algorithm for  Online Facility Location}\label{algo:FL-gen}
\end{algorithm}

\begin{proof}

    In Phase~1 of the algorithm, we pay cost at most $O(1)\cdot \E \opt(S)$, so we only need to bound the cost incurred in Phase~2. For each $x_i \in R$, we define $\cost(x_i) := f \cdot \ind{\text{facility opened at } x_i} + d(x_i, F_{i})$, where $F_i$ denotes the facilities open after the $i$th real arrival.

    Consider the optimal clusters $C_1, \dots, C_m$ in the offline optimal solution $\opt(V)$, and let $c_1, \dots, c_m$ be the respective optimal facility for each cluster. For each $k \in [m]$, let 
    \[r_k := \frac{1}{|C_k|} \sum_{v \in C_k} d(v, c_k)
    \]
    denote the average assignment distance of a facility in $C_k$. We partition the $k$th cluster into $L = \Theta(\log(1/p))$ \emph{rings} $C_k^0, C_k^1, \dots, C_k^L$, along with an outer component $C_k^{out}$, given by 
    \begin{align*}
        C_k^0 &:= \{v \in C_k : d(v, c_i) \leq r_k\},\\
        C_k^\ell &:= \{v \in C_k : e^{\ell -1}r_k \leq d(v, c_i) \leq e^{\ell}r_k\},  \qquad \forall \ell \in [L],\\
        C_k^{out} &:= \{v \in C_k : e^Lr_k \leq d(v, c_i)\}.
    \end{align*}
    For convenience, we will also use $C_k^{in} := C_k \setminus C_k^{out} = C_k^0 \cup \dots \cup C_k^L$.

    \begin{claim}\label{clm:mrf-fl-inner}
        For every cluster $k \in [m]$, we have
        \begin{align*}
            \E \sum_{x_i \in R \cap C_k^0} \cost(x_i) &\leq 2f + 2\sum_{x \in C_k^0} d(c_k, x) + 2\sum_{x \in C_k} d(c_k, x),\\
            \E \sum_{x_i \in R \cap C_k^\ell} \cost(x_i) &\leq 2f + 6\sum_{x \in C_k^\ell} d(c_k, x), && \forall \ell \in [L].
        \end{align*}
    \end{claim}
    \begin{proof}[Proof of \Cref{clm:mrf-fl-inner}]
        This follows directly from the analysis of Meyerson, Lemma 2.1 and Theorem 4.2.
    \end{proof}
    
    \begin{claim}\label{clm:mrf-fl-outer}
        For each $k \in [m]$, we have
        \begin{align*}
        \E \sum_{x_i \in R \cap C_k^{out}} \cost(x_i) &\leq 4\sum_{x \in C_k^{out}} d(c_k, x) + 2\E\sum_{x \in S\cap C_k^{in}} d(x, \hat F) +\E\Big(|C_k^{out}| - |S \cap C_k^{in}|\Big)^+ \cdot f.
        \end{align*}
    \end{claim}
    \begin{proof}[Proof of \Cref{clm:mrf-fl-outer}] To show this claim, we will attempt to match as many elements of $R\cap C_k^{out}$ (the real arrivals in the outermost ring) as possible with the elements of $S \cap C_k^{in}$ (the sample elements in any inner ring), with the hope that not too many elements of the former set go unmatched. We will then bound the cost of demands in $R\cap C_k^{out}$ by the cost of the corresponding elements of $S\cap C_k^{in}$ to the sample solution.
    
    Formally, let $A \subseteq R\cap C_k^{out}$ be a subset of size $\min \{|R\cap C_k^{out}|,~|S\cap C_k^{in}|\}$, and let $g : A \to S\cap C_k^{in}$ be an injective map. For this proof, we also condition on the choice sample and real elements. For each $x_i \in A$,  
    \begin{align*}
        \E[ \cost(x_i) \mid S] &\leq 2d(x_i, \hat F)\\
        &\leq 2d(x_i, c_k) + 2d(c_k, g(x_i)) +2d(g(x_i), \hat F)\\
        &\leq 4d(x_i, c_k) + 2d(g(x_i), \hat F).
    \end{align*}
    In addition, for each demand $x \in (R\cap C_k^{out})\setminus A$, we incur cost at most $f$. 
    
    Summing these bounds over all $x \in R \cap C_k^{out}$ and taking expectation over  $S$ gives our claim.
    \end{proof}
    \begin{claim}\label{clm:mrf-fl-chernoff}
    For each $k \in [m]$, we have $\E\Big(|C_k^{out}| - |S \cap C_k^{in}|\Big)^+ \leq 2.$
    \end{claim}
    \begin{proof}
        Let $z := |C_k^{out}|$. Notice that $z \leq O(p) \cdot |C_k|$ by construction of $C_k^{out}$ and Markov's inequality, so we have $|C_k^{in}| \geq (1-p)|C_k|$. Since $\Pr(x \in S) = p$ independently for each $x \in C_k^{in}$, we have that $|S \cap C_k^{in}|$ stochastically dominates the binomial distribution $Bin(|C_k^{in}|,p)$. Thus, by choosing $L = O(\log(1/p))$ (i.e. so that $\E|S\cap C_k^{in}| \geq p(1-p)|C_k| \geq 2z$) and using Chernoff bounds, we have
        \begin{align*}
            \Pr(|S \cap C_k^{in}| < z) \leq e^{-z/4}.
        \end{align*}
        Therefore, we conclude 
        \[\E\Big(|C_k^{out}| - |S \cap C_k^{in}|\Big)^+ \leq |C_k^{out}| \cdot \Pr(|S \cap C_k^{in}| < |C_k^{out}|) \leq z \cdot e^{-z/4} \leq 2.   \qedhere \]
    \end{proof}
    Putting together the last three claims proves \Cref{thm:facilLocMonot}.
\end{proof}
